\pdfoutput=1
\documentclass[a4 paper]{article}
\usepackage{fancyhdr}
\usepackage{latexsym}
\usepackage{mathrsfs}
\usepackage{cancel}
\usepackage{color}
\usepackage{multirow}
\usepackage{eso-pic}
\usepackage{dsfont}
\usepackage{amssymb}
\usepackage{graphicx}
\usepackage{amsmath}
\usepackage{setspace}
\usepackage{geometry}
\usepackage{amsthm}

\usepackage{amssymb,amsfonts}
\usepackage[all,arc]{xy}
\usepackage{enumerate}
\usepackage{mathrsfs}
\usepackage{bbm}

\makeatletter
\def\cleardoublepage{\clearpage\if@twoside \ifodd\c@page\else
   \hbox{}\thispagestyle{empty}\newpage\addtocounter{page}{-1}
   \if@twocolumn\hbox{}\newpage\fi\fi\fi}
\makeatother

\makeatletter
\newcommand{\colim@}[2]{%
  \vtop{\m@th\ialign{##\cr
    \hfil$#1\operator@font colim$\hfil\cr
    \noalign{\nointerlineskip\kern1.5\ex@}#2\cr
    \noalign{\nointerlineskip\kern-\ex@}\cr}}%
}
\newcommand{\colim}{%
  \mathop{\mathpalette\colim@{\rightarrowfill@\scriptscriptstyle}}\nmlimits@
}
\renewcommand{\varprojlim}{%
  \mathop{\mathpalette\varlim@{\leftarrowfill@\scriptscriptstyle}}\nmlimits@
}
\renewcommand{\varinjlim}{%
  \mathop{\mathpalette\varlim@{\rightarrowfill@\scriptscriptstyle}}\nmlimits@
}
\makeatother

\usepackage{graphicx}
\DeclareSymbolFont{lettersA}{U}{txmia}{m}{it}
\DeclareMathSymbol{\pi}{\mathord}{lettersA}{25}
\DeclareMathSymbol{\muup}{\mathord}{lettersA}{22}

\usepackage{braket}
\usepackage{enumitem}
\usepackage[backend=biber, style=alphabetic, maxbibnames=99, minalphanames=10, maxalphanames=99]{biblatex}
\renewbibmacro{in:}{}

\usepackage{url}

\usepackage[breaklinks]{hyperref}
    \hypersetup{
        colorlinks=true,
        linktoc=all,
        linkcolor=blue,
        breaklinks=true,
        pdftitle={3d Mirror Symmetry and the beta-gamma VOA},
		pdfauthor={Andrew Ballin and Wenjun Niu},
    }
\usepackage{tikz-cd}
\usepackage{authblk}

\usetikzlibrary{decorations.pathmorphing}

\newtheorem{defn}{Definition}[section]
\newtheorem{theorem}[defn]{Theorem}

\newcommand{\bg}{{\beta\gamma}}
\newcommand{\normord}[1]{:\mathrel{\mspace{2mu}#1\mspace{2mu}}:}
\newcommand{\mbb}{\mathbb}
\newcommand{\unit}{\mathds{1}}
\newcommand{\mc}{\mathcal}
\newcommand{\cbg}{{\mathcal{C}_{\beta \gamma}}}
\newcommand{\inv}{^{-1}}
\renewcommand{\th}{^\text{th}}

\newcommand{\gl}{{\mathfrak{gl}(1|1)}}
\newcommand{\cp}{\mathbb{CP}}
\newcommand{\map}[3]{\mbox{$#1 \colon #2 \to #3$}}
\newcommand{\geig}{\text{geig}}
\newcommand{\ol}{\overline}
\newcommand{\ext}{\text{Ext}}
\renewcommand{\hom}{\text{Hom}}

\makeatletter
\newcommand*{\defeq}{\mathrel{\rlap{%
                     \raisebox{0.3ex}{$\m@th\cdot$}}%
                     \raisebox{-0.3ex}{$\m@th\cdot$}}%
                     =}
\makeatother

\begin{document}
\title{3d Mirror Symmetry and the $\beta\gamma$ VOA}
\author[1]{Andrew Ballin\thanks{\href{mailto:asballin@ucdavis.edu}{asballin@ucdavis.edu}}}
\author[2]{Wenjun Niu\thanks{\href{mailto:wjniu@ucdavis.edu}{wjniu@ucdavis.edu}}}
\affil[1]{Department of Physics \& Astronomy and Center for Quantum Mathematics and Physics (QMAP), University of California, Davis, CA 95616, USA}
\affil[2]{Department of Mathematics and Center for Quantum Mathematics and Physics (QMAP), University of California, Davis, CA 95616, USA}
\renewcommand\Affilfont{\itshape\small}

\maketitle

\numberwithin{equation}{section}
\begin{abstract}

We study the simplest example of mirror symmetry for 3d $\mc N=4$ SUSY gauge theories: the A-twist of a free hypermultiplet and the B-twist of SQED. We particularly focus on the category of line operators in each theory. Using the work of Costello-Gaiotto, we define these categories as appropriate categories of modules for the boundary vertex operator algebras present in each theory. For the A-twist of a free hyper, this will be a certain category of modules for the $\beta\gamma$ VOA, properly containing the category previously studied by Allen-Wood. Applying the work of Creutzig-Kanade-McRae and Creutzig-McRae-Yang, we show that the category of line operators on the A side possesses the structure of a braided tensor category, extending the result of Allen-Wood. In addition, we prove that there is a braided tensor equivalence between the categories of line operators on the A side and B side, completing a non-trivial check of the 3d mirror symmetry conjecture. We derive explicit fusion rules as a consequence of this equivalence and obtain interesting relations with associated quantum group representations.

\end{abstract}

\tableofcontents

\section{Introduction}\label{1}


Given a reductive Lie group $G$ and a representation $V$ of $G$, physicists have defined a 3d $\mc N=4$ supersymmetric gauge theory $T_{G,V}$ with gauge group $G$ and hypermultiplets in the representation $V\oplus V^*$ of $G$. For each nilpotent element in the SUSY algebra, one can \emph{twist} the theory by taking cohomology of the operators with respect to this element. In this paper, we focus on the simplest examples of such theories: the A-twist of one free hypermultiplet $T_A$ and the B-twist of SQED $T_B$. We carefully define and study the category of line operators in the theories and show that they have the structure of a braided tensor category. Moreover, we prove that the category of line operators for the two theories are equivalent as braided tensor categories, completing a significant check on a conjectured 3d mirror symmetry. For original works on 3d mirror symmetry, we refer the reader to \cite{is, de1997mirror, de1997mirror2}. For details specifically pertaining to line operators in these theories, see \cite{dimofte2020mirror}.

Given the theories $T_A$ and $T_B$, one can consider their respective category of line operators ${\mc L}_A$ and ${\mc L}_B$. Objects in these categories are operators supported on a line, and the Hom space between two line operators is the space of local operators that can sit at the junction of the two lines. Since the twist element is topological (namely translations are exact under the twist), one should be able to bring two line operators close to each other, forming a new line operator. One should also be to wind one line around the other. These observations suggest that ${\mc L}_A$ and ${\mc L}_B$ should have the structure of a braided tensor category.

To appropriately define and study the structure of these categories, we put these theories on a manifold with boundary. This is a direct extension of Witten's approach of braiding in 3d Chern-Simons theory to the setting of supersymmetric gauge theories \cite{wittenJonesPoly}. In \cite{costello2019vertex}, the authors studied various boundary conditions of 3d $\mc N=4$ theories that are compatible with the topological twist in the bulk. The twist is holomorphic at the boundary, and so the boundary supports a vertex operator algebra (VOA). The authors identified these VOAs for various theories including the pair of theories $T_A$ and $T_B$ that we study here. Given a line operator in the bulk, one can position it so that it ends on the boundary. Bringing a boundary field close to the endpoint of the line gives an action of the boundary VOA on the line. Thus, for each choice of boundary condition, one gets a functor
\begin{equation}
\map{\mc T}{\mc L}{D^b{\mc V}_{\text{bdry}}\!\!-\text{Mod}}.
\end{equation}
Here ${\mc L}$ is the category of line operators in the bulk and $D^b{\mc V}_{\text{bdry}}\!\!-\text{Mod}$ denotes the (derived) category of modules of the boundary VOA ${\mc V}_{\text{bdry}}$. The functor $\mc T$ is expected to be a derived functor and a tensor functor. In general, $\mc T$ may not be fully faithful and conservative\footnote{Here conservative means that the functor maps a non-zero object to a non-zero object.}. However, one expects $\mc T$ to enjoy these properties for ``good" choices of boundary conditions, and will assume this to be the case throughout this paper; see \cite{costello2019vertex} and \cite{costello2019higgs} for further discussion. For ``good" boundary conditions, we may use a suitable (derived) category of modules of the VOA ${\mc V}_{\text{bdry}}$ as a candidate for ${\mc L}$, and fusion of line operators corresponds to fusion product of modules of the boundary VOA. We comment that throughout this paper, we will focus on the \textbf{abelian} category of VOA modules, with the exception of Section \ref{5.1} in which we compute the derived endomorphism space of the identity line operator in $D^b\cbg$.

For $T_A$, we impose a Neumann boundary condition on the hypermultiplet. By \cite{costello2019vertex}, the boundary vertex operator algebra can be identified with the symplectic boson VOA, also known as the $\beta\gamma$ VOA $V_{\beta\gamma}$. The mode algebra of this VOA is identified with the algebra of differential operators on the loop space $\mathcal K$. In Section \ref{bgrepcat}, we define a category of modules of $V_{\beta\gamma}$, which we denote by $\cbg$. We propose $D^b\cbg$ is equivalent to ${\mc L}_A$. In \cite{rw} and \cite{aw}, the authors studied a smaller module category and showed that it has the structure of a braided tensor category defined by $P(z)$-interwining operators, using the machinery of \cite{hlz1}-\cite{hlz8}. However, their category is not the correct one for our setting; in particular, it is too small to reproduce the correct Coulomb branch, which was computed in \cite{costello2019vertex, costello2019higgs}. On the other hand, once we enlarge the category to $\cbg$, it is very difficult to verify that $\cbg$ satisfies the conditions of Huang-Lepowsky-Zhang and to compute the fusion structure. To overcome this, we will use the idea of 3d mirror symmetry to relate $V_{\beta\gamma}$ to the VOA appearing for the ``mirror dual boundary condition" on the mirror theory $T_B$.

The boundary condition we impose on $T_B$ will be the Dirichlet boundary condition for the hypermultiplet. In \cite{costello2019higgs}, at least perturbatively, the boundary VOA is identified with the vertex operator superalgebra (VOSA) $V(\widehat \gl)$, which is the VOSA associated to the affine Lie superalgebra $\widehat \gl$. In order to take into account the contribution from monopole operators, one needs to take an extension of $V(\widehat \gl)$. This is similar to the story of extending affine $ \mathfrak{gl}_1$ Kac-Moody to WZW modules. In \cite{creutzig2013w}, the authors considered various extensions of $V(\widehat \gl)$ by simple currents $\widehat{A}_{n,l}$ and identified them as known VOSAs. By comparing indices of the vacuum modules of \cite{creutzig2013w} with the prediction from \cite{dimofte2018dual}, we find that the VOSA extension generated by $\widehat{A}_{0,\pm 1}$ has index matching the physical prediction. The extended VOSA, which we denote by ${\mc V}_{ext}$, is identified in \cite{creutzig2013w} as the VOSA $V_{\beta\gamma}\otimes V_{bc}$, the tensor product of $V_{\beta\gamma}$ with the $bc$-ghost VOSA\footnote{The $bc$-ghost system is simply a pair of free fermions. These fermions do not physically interact with the $\bg$ VOA in our paper, hence tensoring with the $bc$-ghost algebra will not modify our representation category $\cbg$}. This justifies the name ``mirror dual boundary condition".

We can now introduce the candidate for ${\mc L}_B$. In \cite{creutzig2020tensor}, the authors studied the Kazhdan-Lusztig category $KL$ of $V(\widehat \gl)$ and showed that it is a rigid braided tensor supercategory defined by $P(z)$-intertwining operators. In Section \ref{4}, we will introduce a tensor subcategory $KL^0$. The simple currents $\widehat{A}_{0,\pm 1}$ generate an action of $\mathrm{Coh}(\mathbb{Z})\cong \mathrm{Rep}(\mathbb{C}^*)$ on $KL^0$. Let $KL^0/\mathbb{Z}$ be the de-equivariantization of $KL^0$ by this action\footnote{This notation will be explained in Section \ref{4}.}. De-equivariantizing can be understood in physical terms as taking into account the monopole operators present in the theory. We can further de-equivariantize by the action of $\mathrm{Coh}(\mathbb{Z}_2)$, generated by parity shift, to obtain the category $KL^0/(\mathbb{Z}\times \mathbb{Z}_2)$. We propose that ${\mc L}_B=D^b\Big( KL^0/(\mathbb{Z}\times \mathbb{Z}_2) \Big)$. With this definition, the work of \cite{creutzig2020tensor} immediately implies that ${\mc L}_B$ is a braided tensor category.\footnote{Note that we choose to define the category using $V(\widehat \gl)$ instead of the extended VOA because the category for the former is better known \cite{creutzig2020tensor}.}

Our main result of the paper is the following:

\newtheorem{LALBequiv}{Theorem}[section]

\begin{LALBequiv}\label{LALBequiv}
The category $\cbg$ has the structure of a braided tensor category defined by $P(z)$-intertwining operators. Moreover, there is an equivalence between braided tensor categories:
\begin{equation}
\cbg\cong KL^0/(\mathbb{Z}\times \mathbb{Z}_2).
\end{equation}
\end{LALBequiv}
\noindent Upon taking the derived categories on both sides, we obtain ${\mc L}_A\cong {\mc L}_B$.

The strategy of the proof is to use the fact that ${\mc V}_{ext}$ is an extension of $V(\widehat \gl)$ isomorphic to $V_{\beta\gamma}\otimes V_{bc}$. In \cite{creutzig2020direct}, the authors explained that in this case, the category
\begin{equation}
\mathrm{Rep}^0({\mc V}_{ext})={\mc V}_{ext}-\mathrm{Mod}(\mathrm{Ind}(KL))
\end{equation}
has the structure of a braided tensor supercategory defined by $P(z)$-intertwining operators. Here $\mathrm{Ind}(KL)$ is certain ind-completion of $KL$ \cite{creutzig2020direct}. Moreover, there is a tensor functor $\map{\mc F}{KL^0}{\mathrm{Rep}^0({\mc V}_{ext})}$, identifying the image of $KL^0$ with the de-equivariantization $KL^0/\mathbb{Z}$. On the other hand, we will carefully analyze the category $\cbg$ and obtain classification results in Theorem \ref{cbgindecomptyp} and Theorem \ref{cbgindecompatyp}. One main difficulty in proving these theorems is the lack of injectives and projectives in our category. Of particular note is our classification of the indecomposable objects that were absent in the categories of \cite{rw,aw}. We use these classification results, together with an application of free-field realizations, to show that the image of ${\mc F}$ coincides with $\cbg\boxtimes \mathrm{SVect}$:\footnote{The $\text{SVect}$ factor appears because we have chosen to regard $V_{bc}$ as a vertex operator \emph{super}algebra for mathematical convenience; de-equivariantizing by the action of $\text{Coh}(\mbb Z_2)$ serves to eliminate this superfluous information so that we end up obtaining $\cbg$ on the nose.}
\begin{equation}\label{cbKL0/Z}
   KL^0/\mathbb{Z}\cong {\mc F}(KL^0)= \cbg\boxtimes \mathrm{SVect}.
\end{equation}
This will automatically imply that $\cbg\boxtimes \mathrm{SVect}$ has the structure of a braided tensor supercategory defined by $P(z)$-intertwining operators, and that the equivalence in equation \eqref{cbKL0/Z} is an equivalence between braided tensor supercategories. From this equivalence we can easily deduce Theorem \ref{LALBequiv}. Along the way of the proof, we obtain a nice criteria for objects in $KL$ to be in $KL^0$ in terms of the action of certain element of $\gl$. We comment that the idea of relating $V(\widehat\gl)$ and $V_{\beta\gamma}$ to study the fusion structure of $V_{\beta\gamma}$ modules has previously appeared in the work of \cite{adamovic2019fusion}, wherein fusion product formulae in the smaller category of \cite{rw} were obtained using fusion products of $V(\widehat\gl)$ modules.\footnote{We thank Adamovi{\'c} for pointing out to us this previous work.}

The two main difficulties involved in proving Theorem \ref{LALBequiv} are: to understand $\cbg$ and $KL^0$, and to understand the lifting functor ${\mc F}$. The first difficulty is tackled in Section \ref{2}, especially in Theorem \ref{cbgindecomptyp} and Theorem \ref{cbgindecompatyp}, where we give a classification result for the objects in $\cbg$ by carefully computing extension groups between objects. The second difficulty is overcome in Section \ref{4} through an application of \cite{creutzig2020simple, creutzig2020direct} together with free-field realizations described in \cite{creutzig2020duality,aw}. 

We comment that the equivalence in Theorem \ref{LALBequiv} is related to the 3d mirror symmetry statement in the work of Hilburn-Raskin \cite{hilburn2021tate}. Indeed, as mentioned above, the mode algebra of $V_{\beta\gamma}$ is equivalent to the algebra of differential operators on $\mathcal K$, and the category $\mathrm{D}^!\mathrm{-mod}(\mc K)$ is taken as the A-side category in \cite{hilburn2021tate}.  Our category $\cbg$ is a full subcategory of $\mathrm{D}^!\mathrm{-mod}(\mc K)$.\footnote{Intuitively, one can think of ${\mc L}_A$ as the subcategory of twisted-monodromic modules strongly equivariant with respect to the Iwahori subgroup of $\mathcal{O}^*$.} The B-side category of \cite{hilburn2021tate} is taken to be $\mathrm{IndCoh}(\mathcal Y)$, where $\mc Y$ is roughly the moduli spaces of triples $( L, \nabla, s)$ where $L$ is a line bundle over the punctured disk $\mathbb{D}^*$, $\nabla$ is a connection, and $s$ is a flat section. By an optimistic Koszul duality, $\mathrm{IndCoh}(\mc Y)$ can be related to the category of modules of $\widehat \gl$.\footnote{Intuitively, one can think of ${\mc L}_B$ as the subcategory of coherent sheaves supported on the subspace of $\mc Y$ where $\nabla$ has a regular-singularity.} At the level of derived categories, the equivalence of Theorem \ref{LALBequiv} should be a restriction of the following statement \cite{hilburn2021tate}:

\newtheorem{HilburnRaskin}[LALBequiv]{Theorem}

\begin{HilburnRaskin}[Hilburn-Raskin]
There is an equivalence:
\begin{equation}
\mathrm{D}^!\mathrm{-mod}(\mc K)\cong \mathrm{IndCoh}(\mc Y).
\end{equation}
\end{HilburnRaskin}
The categories in the above statement have very little chance of being braided tensor categories, but rather are expected to have the structure of chiral categories \cite{raskin2015chiral}. In physics terms, objects in the above categories may have nontrivial OPE with each other or with local operators. Restricting to appropriate VOA categories allows fusion and braiding, which is the advantage of our approach. 


This paper is structured in the following way. In Section \ref{2}, we give a review of the vertex algebra $V_{\beta\gamma}$, define our module subcategory $\cbg$, and obtain classification results for $\cbg$. In Section \ref{3}, we give a review of the vertex algebra $V(\widehat \gl)$ and the category of modules $KL$, and compute fusion products of indecomposable objects in $KL$. In Section \ref{4}, we prove Theorem \ref{LALBequiv}, and use this to compute fusion rules for indecomposable objects in $\cbg$. In Section \ref{5}, we relate the subcategory of $\cbg$ of atypical modules to the category of representations of a quiver algebra and quantum group, and compute the endomorphism algebra of the identity line operator.

\vspace{1em}
\noindent\textbf{Acknowledgements.} We thank Thomas Creutzig for teaching us the relation between $V(\widehat \gl)$ and $V_{\beta\gamma}$ and sharing with us many insights into the subject of this paper. We thank Tudor Dimofte for suggesting this research project and providing much guidance along the way. We thank Robert Allen, Niklas Garner, Nathan Geer, Justin Hilburn, David Ridout, and Simon Wood for many helpful discussions.

\section{The $\beta\gamma$ VOA and its representation category}\label{2}

In this section, we focus on the $\bg$ vertex operator algebra $V_{\beta\gamma}$. In Section \ref{2.1}, we recall the definition of $V_{\beta\gamma}$. In Section \ref{bgrepcat}, we recall basic examples of modules of $V_{\beta\gamma}$ following \cite{aw}, and define the category of interest $\cbg$. In Section \ref{2.3}, we give examples of indecomposable modules in $\cbg$, and prove classification results in Theorem \ref{cbgindecomptyp} and Theorem \ref{cbgindecompatyp}. In Section \ref{2.4}, we recall the notion of $P(z)$-intertwining operators and explain the difficulty in constructing monoidal structure for $\cbg$.

\subsection{Definitions and conventions}\label{2.1}
In this paper, we define a vertex operator algebra (VOA) to be a $\mbb Z$-graded vector space $V$ together with a state-operator correspondence $\map{ Y }{ V }{ \text{End}(V) [\![z,z\inv]\!] }$ and a conformal element $\omega \in V$, both subject to various conditions; see the definition of a conformal vertex algebra in \cite[Defintion 2.2]{hlz1} for complete details. Our definition of a vertex operator superalgebra (VOSA) closely follows that of a VOA but there is an additional $\mbb Z_2$ grading present and all VOA structures and conditions are rephrased to make them compatible (e.g. commutativity conditions are replaced by \emph{graded} commutativity conditions); see remark 1 of \cite[Section 1.3.2]{frenkel2004vertex} and \cite[Section 1.4]{creutzig2017tensor} for further details.

The $\bg$ vertex operator algebra $V_{\beta\gamma}$ is strongly generated by two bosonic fields
\begin{equation}
    \beta(z) = \sum_{n \in \mbb Z} \beta_n z^{-n-1} \qquad \gamma(z) = \sum_{n \in \mbb Z} \gamma_n z^{-n}
\end{equation}
satisfying the operator product expansions
\begin{equation}
    \beta(z) \beta(w) \sim 0 \qquad 
        \gamma(z) \gamma(w) \sim 0 \qquad
        \beta(z) \gamma(w) \sim \frac{-1}{z-w}
\end{equation}
It admits the structure of a $\mbb Z_{\ge 0}$-graded VOA when equipped with the following choice of conformal element/stress-energy tensor \cite{rw}
\begin{equation}
    \omega = -\beta_{-1}\gamma_{-1} \qquad T(z) = - \normord{\beta(z) \partial \gamma(z)} = \sum_{\ell \in \mbb Z} z^{-\ell -2} \left [ \sum_{k \in \mbb Z} k \normord{\beta_{\ell-k} \gamma_k} \right ]
\end{equation}
The OPEs imply that the coefficients of $\beta(z)$ and $\gamma(z)$ posses the following commutation relation
\begin{equation}
    [\beta_m, \beta_n] = 0 \qquad [\gamma_m, \gamma_n] = 0 \qquad [\beta_m, \gamma_n] = - \delta_{m,-n}\unit
\end{equation}
For each $n \in \mbb Z$, $\bg_n \defeq \mbb C[ \beta_n, \gamma_{-n}]$ is a 1-dimensional Weyl algebra under the identification $\beta_n \leftrightarrow x_{-n}$ and $\gamma_n \leftrightarrow \partial_n$. Our analysis of the category of VOA modules strongly relies on this simple observation. We denote the universal enveloping algebra generated by $\{ \beta_n, \gamma_n \}_{n \in \mbb Z}$, which can be thought of as an infinite-dimensional Weyl algebra, by $\bg$.

The $\bg$ VOA possesses an additional global $U(1)$ symmetry (i.e. contains a $U(1)$ Kac-Moody VOA) whose associated current is
\begin{equation}
    J(z) = \, \normord{\beta(z) \gamma(z)} \, = \sum_{\ell\in \mbb Z} z^{-\ell -1} \left [ \sum_{k \in \mbb Z} \normord{\beta_k \gamma_{\ell-k}}\right ]
\end{equation}
In addition to the $\mbb Z_{\ge 0}$-grading given by $L_0$, the $\bg$ VOA is strongly $\mbb Z$-graded (in the sense of \cite[Definition 2.23]{hlz1}) with respect to $J_0$.

It will be handy for future computations to record some of the commutation relations and modes of $T(z)$ and $J(z)$
\begin{align}
    J_0 
        & = \sum_{n\ge 0} \gamma_{-n}\beta_n + \sum_{n \ge 1} \beta_{-n} \gamma_n \label{j0} \\
    L_0
        & = \sum_{k \ge 1} k \left [ \beta_{-k} \gamma_k - \gamma_{-k} \beta_k\right ] \label{l0} \\
    L_{-1}
        & = \sum_{k \ge 1} k \left [\beta_{-1-k} \gamma_k - \gamma_{-k} \beta_{k-1} \right ]
\end{align}
\begin{equation}\label{bgJLcomms}
\begin{split}
    [ J_0, \beta_k ] = \beta_k
        & \qquad [ L_0, \beta_k ] = -k \beta_k \\
    [ J_0, \gamma_k ] = -\gamma_k
        & \qquad [ L_0, \gamma_k ] = -k\gamma_k
\end{split}
\end{equation}
The following $\bg$ automorphisms will frequently appear when discussing modules for $V_\bg$:
\begin{itemize}
\item Conjugation: $c(\beta_n) = \gamma_n \qquad c(\gamma_n) = -\beta_n$
\item Spectral flow: $\sigma(\beta_n) = \beta_{n-1} \qquad \sigma(\gamma_n) = \gamma_{n + 1}$
\end{itemize}
When combined with the $U(1)$ global symmetry described above, the existence of the conjugation automorphism tells us that $V_\bg$ actually has an $Sp(2)$ global symmetry. The spectral flow automorphism can be thought of as arising from a 1-form symmetry present in the 3d bulk theory (with a line operator) whose boundary algebra is $V_\bg$.

\subsection{Our large representation category}\label{bgrepcat}
In our paper, a \emph{representation} or \emph{module} of a VOA $V$ is a $\mbb C$-graded vector space $M$, graded by the generalized weights of $L_0$, together with a linear map $\map{ Y_M }{ V }{ \text{End}(M) [\![z, z\inv ]\!] }$ satisfying certain compatibility conditions; see the definition of a generalized module for a conformal vertex algebra in \cite[Definition 2.12]{hlz1} for complete details. Various representation categories $V_\bg$ have been investigated by others (e.g. \cite{rw, aw}). For reasons explained in Section \ref{5.1}, these categories are too small to correctly match the physics. Essentially, the self-extension group (i.e. the derived endomorphism algebra) of the vacuum module in these categories did not match with what one should obtain when computing the bulk local operators. Before we can define the category of physical interest in this paper, we must first introduce some basic modules of the $\bg$ VOA.

Let $\bg_{\ge 0}$ be the unital subalgebra of $\bg$ generated by $\{ \beta_n, \gamma_n, \unit \}_{n \ge 0}$. The simplest module is the vacuum module
\begin{equation}
    \mc V \defeq \text{Ind}^{\bg}_{\bg_{\ge 0}} \mbb C[\gamma_0]
\end{equation}
where $\beta_0$ acts as $-\frac{\partial}{\partial \gamma_0}$ on $\mbb C[\gamma_0]$, and $\beta_n$ and $\gamma_n$ act as 0 for $n \ge 1$. As a vector space, the vacuum module $\mc V$ of $V_{\beta\gamma}$ coincides with $V_{\beta\gamma}$ as a module over itself. 

Similarly, for $\mu \in \mbb C \setminus \mbb Z$, the so called \emph{typical} modules are defined by
\begin{equation}
    \mc W_\lambda \defeq \text{Ind}^{\bg}_{\bg_{\ge 0}}  (\gamma_0)^\mu \mbb C[\gamma_0, \gamma\inv_0]
\end{equation}
where $\beta_0$ acts as $-\frac{\partial}{\partial \gamma_0}$, and $\beta_n$ and $\gamma_n$ act as 0 for $n \ge 1$. Here $\lambda = \mu + \mbb Z$ and one can see that $\mc W_\lambda$ is independent (up to isomorphism) of the choice of $\mu \in \lambda$, hence these modules are parametrized by $(\mbb C \setminus \mbb Z) / \mbb Z$.

Both $\mc V$ and $\mc W_\lambda$ are simple objects. There are two distinct modules that are reducible but indecomposable, and they morally correspond to the different ways one can take the limit of $\mc W_\lambda$ as $\lambda \to \mbb Z$. They are called \emph{atypical} modules and are defined by
\begin{equation}
    \mc W_0^+ \defeq \text{Ind}^\bg_{\bg_{\ge 0}} \mbb C[\gamma_0, \gamma_0\inv]
    \qquad
    \mc W_0^- \defeq \text{Ind}^\bg_{\bg_{\ge 0}} \mbb C[\beta_0, \beta_0\inv]
\end{equation}

Given any $V_\bg$ module $M$, we can construct another module by twisting the action of the $\bg$ VOA with the spectral flow automorphism: for any $n \in \mbb Z$, $\sigma^n M$ is the module that is equal to $M$ as a set, but carries the action
\begin{equation}
    \alpha \star v \defeq \sigma^{-n}(\alpha) \cdot v \quad \text{for every } v \in \sigma^n M, \, \alpha \in \bg
\end{equation}
where $\cdot$ is the action on $M$. For all of the modules above, $\sigma^n M \not \cong \sigma^m M$ for $m \ne n$.

As a side note, the atypicals can equivalently be defined by the Loewy diagrams
\begin{equation}\label{atyploewy}
\mc W_0^+ = \Big( \mc V \longrightarrow \sigma\inv \mc V \Big) \qquad \mc W_0^- = \Big( \sigma\inv \mc V \longrightarrow \mc V \Big)
\end{equation}
In this paper, a Loewy diagram $X \longrightarrow Y$ represents a module that is an extension of $Y$ by $X$. One may have instead drawn such a module as $Y \longrightarrow X$, but we feel our convention better matches the visual appearance of the corresponding short exact sequence describing the extension.

With these examples in hand, we can now define the category that we study in our paper, which we believe properly matches our physical systems of interest.
\begin{defn}\label{catdef}
Let $\cbg$ be the abelian subcategory of smooth, finite-length, $\bg$ VOA modules generated by $\mc V$, $\mc W_\lambda$, and their spectral flows, such that $\cbg$ is closed under taking extension.
\end{defn}



While the element $L_0$ (\ref{l0}) provides the $\mbb C$-grading included in the definition of a VOA module, the element $J_0$ (\ref{j0}) provides an \emph{additional} grading on the modules; see \cite[Definition 2.25]{hlz1} for details. A main feature the reader should keep in mind from Definition \ref{catdef} is that we do not exclude modules on which $L_0$ and $J_0$ act non-semisimply. Additionally, modules in $\cbg$ can be decomposed into a direct sum of finite-dimensional simultaneous generalized eigenspaces for $L_0$ and $J_0$. This should be contrasted with the representation category studied in, e.g., \cite{aw} wherein $J_0$ was required to act semisimply; the category that we study will consequentially be strictly larger. While $\cbg$ is strictly smaller than the category studied in \cite{hilburn2021tate}, the authors do not provide a classification of their category, and more structure exists on $\cbg$ that we study that does not exist on their category. To the best of our knowledge, our present study of $\cbg$ is a new addition to the existing literature on $\bg$ representation categories.

\subsection{Extension structure and classification results}\label{2.3}
To characterize the objects in $\cbg$, we must understand the new modules that are present when we demand closure under extension. For modules induced by polynomial representations of a 1D Weyl subalgebra of $\bg$, one can explicitly construct self-extensions by adjoining powers of a formal variable $\log \beta_0$ or $\log \gamma_0$ before inducing. For example, the first self-extension of $\mc W_0^-$, which we denote by $\mc W_0^{-,2}$, can be constructed by
\begin{equation}
    \mc W_0^{-,2} \defeq \text{Ind}_{\bg_{\ge 0}} (\mbb C[\beta_0, \beta_0\inv] \oplus \mbb C[\beta_0, \beta_0\inv]\log \beta_0)
\end{equation}
where $\gamma_0 \cdot \log \beta_0 = \beta_0\inv$. This module is also is described by the Loewy diagram
\begin{equation}\label{w-2chain}
\begin{tikzcd}
\sigma\inv \mc V \rar
	& \mc V \rar
	& \sigma\inv \mc V \rar
	& \mc V
\end{tikzcd}
\end{equation}
This is an object of $\cbg$ that does \emph{not} carry a semi-simple action of $J_0$, demonstrating that $\cbg$ is an enlargement of the representation category studied in \cite{aw}. We denote the $(n-1)\th$ iterated extension of $\mc W_0^-$ by itself as $\mc W_0^{-,n}$. A similar line of reasoning gives the definition of $\mc W_0^{+,n}$. For example, $\mc W_0^{+,2}$ looks like
\begin{equation}\label{w+2chain}
\begin{tikzcd}
\mc V \rar
    & \sigma\inv \mc V \rar
    & \mc V \rar
    & \sigma\inv \mc V
\end{tikzcd}
\end{equation}
We call these modules ``chains", and by taking submodule/quotient, we can form chains of odd length. We say a chain is \emph{positive} (\emph{negative}) if it is a quotient of some $\sigma^n \mc W_0^{+,k}$ ($\sigma^n \mc W_0^{-,k}$). To the best of our knowledge, these modules have not been studied in the literature.

Since every module can be expressed as a direct sum of indecomposable modules, it suffices to restrict our focus to non-trivial extensions. From the commutation relations in equation \eqref{bgJLcomms}, we see that $\cbg$ admits a block decomposition
\begin{equation}\label{eqcbgblock}
    \cbg = \bigoplus_{\lambda \in \mbb C / \mbb Z} \mc C_{\bg, \lambda}
\end{equation}
where $\mc C_{\bg, \lambda}$ is the full abelian subcategory of $\cbg$ that contains all $\bg$ modules such that the generalized eigenvalues of the representation of $J_0$ lie in $\lambda$. The morphisms and extensions between a module from $\mc C_{\bg, \lambda}$ and another from $\mc C_{\bg, \lambda'}$ are trivial for $\lambda \ne \lambda'$ because morphisms in $\cbg$ respect the generalized eigenvalues of $L_0$ and $J_0$, and equation \eqref{bgJLcomms} tells us that the $\bg$ modes can only shift them by an integer. Therefore we only need to study the extension structure \emph{within} each $\mc C_{\bg, \lambda}$. Modules in $\mc C_{\bg, \lambda}$ for $[\lambda]\ne [0]$ are called typical modules while modules in $\mc C_{\bg, [0]}$ are called atypical modules. 

This task is rather involved, essentially because there are not enough injectives and projectives in $\mc C_\bg$, so we outline our approach before getting into the technical details. It will follow from the definition of $\cbg$ and Lemma $\ref{bgfinext}$ that every module in $\mc C_{\bg, \lambda}$ can be constructed as an induced module of a representation for a finite-dimensional Weyl subalgebra of $\bg$. Consequentially, we will see that the extension structure in $\mc C_{\bg, \lambda}$ can be understood in terms of the extension structure between modules of finite-dimensional Weyl algebras. Therefore we begin by computing the latter, then we state a lemma that explains how to utilize finite-dimensional results to understand the extension structure in $\mc C_{\bg, \lambda}$, and finally we classify the indecomposable objects within each $\mc C_{\bg, \lambda}$. Along the way, we describe a useful way to visualize $\bg$ modules in terms their corresponding finite-dimensional Weyl algebra modules. Concrete examples will also be provided.

As a warm up, we start by computing extensions between some basic modules for the 1D Weyl algebra $H = \mbb C[x, \partial]$ in the category of all $H$-modules. This category corresponds to a $\bg$ VOA module category strictly larger than $\cbg$, but the computations will nonetheless provide us with important information about $\cbg$. In particular, with a bit of extra work, we find that the computations of $\text{Ext}^1$ in this $H$-module category are the same as those in $\cbg$. Let us now begin.


The simple modules $\mbb C[x], \mbb C[\partial],$ and $x^\lambda \mbb C[x, x\inv]$ ($\lambda \in \mbb C \setminus \mbb Z$) each have a 2-step free, hence projective, resolution by $H$. For example,
\begin{equation}\label{freeres}
\begin{tikzcd}
\cdots \rar
    & 0 \rar
    & H \rar{f_1}
    & H \rar{f_0}
    & \mbb C[x] \rar
    & 0
\end{tikzcd}
\end{equation}
is exact, where $f_1(1) = \partial$ and $f_0(1) = 1$. Using these resolutions to compute $\text{Ext}^\bullet$, we arrive at the following results:
\begin{itemize}
    \item $\text{Ext}^k (\mbb C[x], \mbb C[x]) = \text{Ext}^k (\mbb C[\partial], \mbb C[\partial]) = \mbb C \, \delta_{k,0}$. In particular, neither $\mbb C[x]$ nor $\mbb C[\partial]$ have non-trivial self-extensions.
    \item $\text{Ext}^k (\mbb C[\partial], \mbb C[x]) = \mbb C \, \delta_{k,1}$, where the unique non-trivial extension is $\mbb C[x,x\inv]$.
    \item $\text{Ext}^k (\mbb C[x], \mbb C[\partial]) = \mbb C \, \delta_{k,1}$, where the unique non-trivial extension is $\mbb C[\partial, \partial \inv]$.
    \item $\text{Ext}^k(x^\lambda \mbb C[x, x\inv],x^\lambda \mbb C[x, x\inv]) = \mbb C \, \delta_{k,0} \oplus \mbb C \, \delta_{k,1}$ where the unique self-extension is given by $\ol{\mc W}_\lambda^2 \defeq x^\lambda \mbb C[x,x\inv] \oplus x^\lambda \mbb C[x, x\inv] \log x$. Furthermore the unique iterated self-extensions of $x^\lambda \mbb C[x, x\inv]$ are
    \begin{equation}
        \ol{\mc W}_\lambda^k \defeq x^\lambda \mbb C[x, x\inv] \oplus \cdots \oplus x^\lambda \mbb C[x,x\inv] \log^{k-1}x
    \end{equation}
    \item The extension algebra between $\ol{\mc W}_{\lambda_1}^{k_1}$ and $\ol{\mc W}_{\lambda_2}^{k_2}$ is zero for $\lambda_1 \ne \lambda_2$
    \item The extension algebra between $\ol{\mc W}_\lambda^k$ and $\mbb C[x]$ or $\mbb C[\partial]$ are both zero.
\end{itemize}

The following lemma explains why it suffices to compute the extensions between modules by restricting focus to their structure under a finite-dimensional Weyl subalgebra of $\bg$. For this, call
\begin{equation}
    A_N=\mathbb{C}[\beta_k,\gamma_{-k}]_{-N\leq k\leq N}.
\end{equation}

\newtheorem{Lembgfinext}[defn]{Lemma}

\begin{Lembgfinext}\label{bgfinext}
Let $U,V$, and $W$ be in $\mathcal{C}_{\beta\gamma}$, and assume that they fit in the short exact sequence $0\to U\to V\to W\to 0$. Suppose also that both $U$ and $W$ come from induction over $A_N$, namely:
\begin{equation}
    U=\mathrm{Ind}_{A_N[\beta_j,\gamma_j]_{j>N}}^{\beta\gamma}U_N,~~W= \mathrm{Ind}_{A_N[\beta_j,\gamma_j]_{j>N}}^{\beta\gamma}W_N.
\end{equation}
Then $V$ also comes from induction of an $A_N$-module $V_N$, and the above short exact sequence comes from the induction of the short exact sequence:
\begin{equation}
\begin{tikzcd}
0 \rar & U_N \rar & V_N \rar & W_N \rar & 0.
\end{tikzcd}
\end{equation}
\end{Lembgfinext}

\begin{proof}
Given a module $M$ of $\beta\gamma$, denote by $K_N(M)$  the kernel of all the $\beta_k$ and $\gamma_k$ for $k>N$. Such $K_N(M)$ is easily seen to be a module of $A_N$. When $M$ comes from induction from some $A_N$-module $M_N$, we have $K_N(M)=M_N$. Applying this to the short exact sequence $0\to U\to V\to W\to 0$, using the fact that taking kernel is left exact, we get:
\begin{equation}
\begin{tikzcd}
0 \rar
    & U_N \rar
    & K_N(V) \rar
    & W_N.
\end{tikzcd}
\end{equation}
We claim that $K_N(V)\to W_N$ is surjective. Given any $w\in W_N$, since $V\to W$ is onto, we may choose $v\in V$ such that it's image in $W$ is $w$. Since $V$ is a generalized VOA module, there exists $K$ such that $\beta_kv=\gamma_kv=0$ for $k>K$. We will adjust $v$ in a way that it's image in $W$ is still $w$, but it will be annihilated by $\beta_j$ and $\gamma_j$ for $N<j\leq K$. For any such mode, say $\beta_j$, if $\beta_j v\ne 0$, then since it's image in $W$ is zero, it must be in the kernel of $V\to W$, which is $U$. Thus $\beta_jv\in U$. By the fact that $U$ comes from induction, there exists $n$ such that $\beta_j^{n+1}v=0$. Now using $[\beta_j^n,\gamma_{-j}]=-n\beta_j^{n-1}$, one gets:
\begin{equation}
    \beta_j^n\gamma_{-j}\beta_j v=-n\beta_j^{n}v.
\end{equation}
In other words, $\beta_j^n (v+\frac{1}{n}\gamma_{-j}\beta_j v)=0$. The element $v+\frac{1}{n}\gamma_{-j}\beta_j v$ has the same image as $v$ in $W$, since $\gamma_{-j}\beta_j v\in U$. Inductively, one can find $u\in U$ such that $v-u$ is in fact annihilated by $\beta_j$, and $u$ is of the form $f(\beta_j,\gamma_{-j})\beta_jv$ for some polynomial $f$ of two variables, where we always choose the normal ordering in the polynomial, namely $\gamma_{-j}$ appears before $\beta_{j}$. Notice that if $v$ is annihilated by $\beta_k$ for $k\ne j$, then so is $v-f(\beta_j,\gamma_{-j})\beta_jv$. Applying this to all the $\beta_j$ and $\gamma_j$ for $N<j<K$, one gets $v$ that is annihilated by all $\beta_j$ and $\gamma_j$ for $j>N$, namely $v\in K_N(V)$ whose image is $w$. 

In conclusion, we have a short exact sequence:
\begin{equation}
\begin{tikzcd}
0 \rar
    & U_N \rar
    & K_N(V) \rar
    & W_N \rar
    & 0
\end{tikzcd}
\end{equation}
Induction gives us a short exact sequence that fits in the diagram:
\begin{equation}
\begin{tikzcd}
    0 \rar
        & U \rar \dar
        & \mathrm{Ind}_{A_N[\beta_j,\gamma_j]_{j>N}}^{\beta\gamma}( K_N(V)) \rar \dar
        & W \rar \dar
        & 0 \\
    0 \rar
        & U \rar
        & V \rar
        & W \rar 
        & 0
\end{tikzcd}
\end{equation}
The left and right down-arrows are isomorphisms, so by the snake lemma, the middle down-arrow is an isomorphism as well. This completes the proof.

\end{proof}

Let us demonstrate how one can combine the 1D results with Lemma \ref{bgfinext} to compute extensions in $\cbg$. Suppose we have a non-trivial extension $M$ of $\mc V$ by $\sigma\inv \mc V$. In the notation of the lemma,
\begin{equation}
    \mc V = \text{Ind}_{A_0[\beta_j, \gamma_j]_{j > 0}}^\bg \mbb C[\gamma_0]
\end{equation}
Similarly,
\begin{equation}
    \sigma\inv \mc V = \text{Ind}_{A_0[\beta_j, \gamma_j]_{j > 0}}^\bg \mbb C[\beta_0]
\end{equation}
By Lemma \ref{bgfinext}, $M$ is the induction of some $A_0[\beta_j, \gamma_j]_{j > 1}$-module $M_0$ that fits into the short exact sequence
\begin{equation}
\begin{tikzcd}
0 \rar
    & \mbb C[\beta_0] \rar
    & M_0 \rar
    & \mbb C[\gamma_0] \rar
    & 0
\end{tikzcd}
\end{equation}
It must be that $M_0$ is a non-trivial extension, since otherwise
\begin{equation}
\begin{tikzcd}
0 \rar
    & \sigma\inv \mc V \rar
    & M \rar
    & \mc V \rar
    & 0
\end{tikzcd}
\end{equation}
would split. From our computational results just below (\ref{freeres}), it must be that $M_0 \cong \mbb C[\beta_0,\beta_0\inv]$, hence
\begin{equation}
    M \cong \text{Ind}_{A_0[\beta_j, \gamma_j]_{j>0}}^\bg \mbb C[\beta_0,\beta_0\inv]
\end{equation}
This module indeed lies in $\cbg$ (i.e. a nontrivial extension of $\mc V$ by $\sigma\inv \mc V$ indeed exists in our category) and it is none other than $\mc W_0^-$ as we expected from equation (\ref{atyploewy})!

Before proceeding to characterize the objects in $\cbg$, we need to investigate how to compute extensions between induced objects a bit further. Suppose we have a module $M \in \cbg$ that comes from the induction of an $A_N$-module $M'$. Then since
\begin{equation}
    \bg \simeq \bigotimes_{k \in \mbb Z} \bg_k = A_N \otimes \left( \bigotimes_{|k|>N} \bg_k \right),
\end{equation}
we have the decomposition
\begin{equation}
    M \cong M' \otimes \left( \bigotimes_{|k|>N} M_k \right)
\end{equation}
where each $M_k$ is a $\bg_k$-module. In fact, each $M_k$ is $\mbb C[\partial_k]$ for $k>0$ and $\mbb C[x_{-k}]$ for $k<0$. We visualize this data in terms of the following ``column picture" for $M$:
\begin{equation}
\begin{array}{c @{\qquad} @{\cdots \quad \otimes \quad} c @{\quad \otimes \quad} c @{\quad \otimes \quad} c @{\quad \otimes \quad} c @{\quad \otimes \quad} c @{\quad \otimes \quad \cdots}}
\bg
	& \bg_{-N-2}
	& \bg_{-N-1}
	& A_N
	& \bg_{N+1}
	& \bg_{N+2}\\
M
	& M_{-N-2}
	& M_{-N-1}
	& M'
	& M_{N+1}
	& M_{N+2}
\end{array}
\end{equation}
Visualizing/decomposing induced modules in this way makes the proofs in this section easier to follow. Let us provide a concrete example of the column picture for $\mc V$. Recalling that $\mc V = \text{Ind}_{\bg \ge 0}^\bg \mbb C[\gamma_0]$ (i.e. $\mc V$ comes from the induction of the $A_0$-module $\mbb C[\gamma_0]$), its column picture is
\begin{equation}
\begin{array}{c @{\qquad} @{\cdots \quad \otimes \quad} c @{\quad \otimes \quad} c @{\quad \otimes \quad} c @{\quad \otimes \quad} c @{\quad \otimes \quad} c @{\quad \otimes \quad \cdots}}
\bg
	& \bg_{-2}
	& \bg_{-1}
	& \bg_0
	& \bg_1
	& \bg_2\\
\mc V 
	& \mbb C[x_2]
	& \mbb C[x_1]
	& \mbb C[\partial_0]
	& \mbb C[\partial_1]
	& \mbb C[\partial_2]
\end{array}
\end{equation}

Note that spectral flow simply shifts the column picture of a module horizontally. For example, the column picture for $\sigma\inv \mc V$ can be obtained by shifting the column picture of $\mc V$ by one unit to the right. We typically relabel the indices on $x_k$ and $\partial_k$ after shifting to make it easier to remember which column they correspond to.

Let us demonstrate how the column picture can be used to compute the extensions between modules with an example. To compute the extensions of $\mc V$ by $\sigma \mc V$, we first stack the column pictures for $\mc V$ and $\sigma \mc V$:
\begin{equation}
\begin{array}{c @{\qquad} @{\cdots \quad \otimes \quad} c @{\quad \otimes \quad} c @{\quad \otimes \quad} c @{\quad \otimes \quad} c @{\quad \otimes \quad} c @{\quad \otimes \quad \cdots}}
\bg
	& \bg_{-2}
	& \bg_{-1}
	& \bg_0
	& \bg_1
	& \bg_2\\
\mc V
	& \mbb C[x_2]
	& \mbb C[x_1]
	& \mbb C[\partial_0]
	& \mbb C[\partial_1]
	& \mbb C[\partial_2]\\
\sigma \mc V
	& \mbb C[x_2]
	& \mbb C[\partial_{-1}]
	& \mbb C[\partial_0]
	& \mbb C[\partial_1]
	& \mbb C[\partial_2]
\end{array}
\end{equation}
Notice that
\begin{equation}
    \sigma \mc V = \text{Ind}_{A_1[\beta_j, \gamma_j]_{j > 1}}^\bg (\mbb C[\partial_{-1}] \otimes \mbb C[\partial_0] \otimes \mbb C[\partial_1])
\end{equation}
where $A_1$ acts in the standard manner as a 3-dimensional Weyl algebra and $\beta_j$ and $\gamma_j$ act as 0 for $j>1$. According to Lemma \ref{bgfinext}, it suffices to instead compute the extensions between the modules contained only in the middle 3 columns
\begin{equation}
\begin{array}{c @{\quad \otimes \quad} c @{\quad \otimes \quad} c}
\bg_{-1}
	& \bg_0
	& \bg_1 \\
\mbb C[x_1]
	& \mbb C[\partial_0]
	& \mbb C[\partial_1] \\
\mbb C[\partial_{-1}]
	& \mbb C[\partial_0]
	& \mbb C[\partial_1]
\end{array}
\end{equation}
Applying a K\"unneth formula, the only non-zero contribution to $\text{Ext}^k(\mc V, \sigma \mc V)$ is in degree 1 and comes from the left column, which is $\mbb C[\partial_{-1}, \partial_{-1} \inv]$. Thus $\text{Ext}^k(\mc V, \sigma \mc V) = \mbb C\, \delta_{k,1}$. The column picture representing the module corresponding to the degree-1 extension is
\begin{equation}
\begin{array}{c @{\quad \cdots \quad \otimes \quad} c @{\quad \otimes \quad} c @{\quad \otimes \quad} c @{\quad \otimes \quad} c @{\quad \otimes \quad} c @{\quad \otimes \quad \cdots}}
\bg
	& \bg_{-2}
	& \bg_{-1}
	& \bg_0
	& \bg_1
	& \bg_2\\
\Big( \sigma \mc V \longrightarrow \mc V \Big)
	& \mbb C[x_2]
	& \mbb C[\partial_{-1}, \partial_{-1} \inv]
	& \mbb C[\partial_0]
	& \mbb C[\partial_1]
	& \mbb C[\partial_2]
\end{array}
\end{equation}
which is precisely $\sigma \mc W_0^+$!\\

Finally, we can characterize the objects of $\cbg$.
\begin{theorem}\label{cbgindecomptyp}
Every indecomposable object in $\mc C_{\bg, \lambda}$ for $\lambda \ne \mbb Z$ is isomorphic to $\sigma^n \mc W_\lambda^k$ for some $n \in \mbb Z$ and $k \in \mbb Z_{\ge 0}$.
\end{theorem}
\begin{proof}
Let $M$ be an indecomposable object in $\mc C_{\bg, \lambda}$ and pick any $\mu \in \lambda$. We induct on the length of $M$. If $M$ is simple, it must be $\sigma^n \mc W_\lambda$ for some $n \in \mbb Z$.

Now suppose $M$ has length $k$ and assume its length $k-1$ submodule (in any particular composition series) is isomorphic to $\sigma^n \mc W_\lambda^{k-1}$. Then $M/\sigma^n W_\lambda^{k-1}$ is a simple module in $\mc C_{\bg, \lambda}$, hence is isomorphic to $\sigma^m \mc W_\lambda$ for some $m \in \mbb Z$. Thus $M$ fits into the short exact sequence
\begin{equation}
\begin{tikzcd}
0 \rar
    & \sigma^n \mc W_\lambda^{k-1} \rar
    & M \rar
    & \sigma^m \mc W_\lambda \rar
    & 0
\end{tikzcd}
\end{equation}
We now characterize all possible extensions of this type. In the column picture of $\sigma^n \mc W_\lambda^{k-1}$, there is a $(\partial_{-n})^\mu ( \mbb C[\partial_{-n}, \partial_{-n}\inv] \oplus \cdots \oplus \mbb C[\partial_{-n}, \partial_{-n}\inv][\log^{k-2} \partial_{-n}])$ in the $\bg_{-n}$ column, and everything in the columns to its left and right are $\mbb C[x]$'s and $\mbb C[\partial]$'s, respectively. The column picture of $\sigma^m \mc W_\lambda$ has a $(\partial_{-m})^\mu \mbb C[\partial_{-m}, \partial_{-m}\inv]$ in the $\bg_{-m}$ column, and the other columns are similarly $\mbb C[x]$'s and $\mbb C[\partial]$'s. Our results about the representation theory of $H$ dictate that $m=n$ in order to have a non-trivial extension. Furthermore, when $m=n$, the same results tell us that the unique non-trivial extension is $(\partial_{-n})^\mu ( \mbb C[\partial_{-n}, \partial_{-n}\inv] \oplus \cdots \oplus \mbb C[\partial_{-n}, \partial_{-n}\inv][\log^{k-1} \partial_{-n}])$. Thus $M \cong \sigma^n \mc W_\lambda^k$, finishing the induction.
\end{proof}

To characterize the indecomposables in $\mc C_{\bg, \mbb Z}$, we introduce a new class of modules, called \emph{roofs}, with the following property: each module in $\mc C_{\bg, \mbb Z}$ can be covered by a finite direct sum of roofs. To construct a roof, one first takes the direct sum of a positive and a negative chain that have the same head, and then one takes the submodule generated by the diagonal of the head. For example, the heads of $\mc W_0^{-,n}$ and $\sigma \mc W_0^{+,m}$ are both $\mc V$, so the roof $\mc R_{2n,2m}$ is the unique submodule of $\mc W_0^{-,n} \oplus \sigma \mc W_0^{+,m}$ generated by the diagonal of the head $\mc V$. The Loewy diagram of $\mc R_{4,4}$, rotated $90^\circ$ clockwise to fit better on the page, looks like
\begin{equation}
\begin{tikzcd}
\sigma\inv \mc V \rar
    & \mc V \rar
    & \sigma\inv \mc V \drar \\
{}
    & 
    & 
    & \mc V \\
\sigma \mc V \rar
    & \mc V \rar
    & \sigma \mc V \urar    
\end{tikzcd}
\end{equation}
which is essentially diagrams (\ref{w-2chain}) and a spectral flow of (\ref{w+2chain}) pinched together at the head. This diagram looks like a tall roof, when unrotated, hence the name. The subscripts $a$ and $b$ on $\mc R_{a,b}$ represent the length of the left and right sides of the roof, respectively. When $a=b$, we drop the redundant subscript
\begin{equation}
    \mc R_a \defeq \mc R_{a,a}
\end{equation}

Our proof of the characterization theorem uses results about the extensions of a chain by $\sigma^m \mc V$. We state the necessary results without proof, but one can easily compute these extension groups with inductive arguments and homological techniques similar to those used below.
\begin{align}
\ext^1(\sigma^n \mc W_0^{+,k}, \sigma^m \mc V)
    & = \mbb C \, \delta_{m,n-2} \oplus \mbb C \, \delta_{m,n-1} \\
\ext^1(\sigma^n \mc W_0^{-,k}, \sigma^m \mc V)
    &= \mbb C \, \delta_{m,n} \oplus \mbb C \, \delta_{m,n+1}
\end{align}

\begin{theorem}\label{cbgindecompatyp}
Every indecomposable object in $\mc C_{\bg, \mbb Z}$ is isomorphic to a quotient of a finite direct sum of $\sigma^n \mc R_k$ for various $n \in \mbb Z$ and $k \in \mbb Z_{\ge 0}$.
\end{theorem}

\begin{proof}
Let $M$ be a length $\ell$ indecomposable object in $\mc C_{\bg, \mbb Z}$. We induct on $\ell$. The statement holds for $\ell = 1$ since all simple modules in $\mc C_{\bg, \mbb Z}$ are $\sigma^n \mc V$, which are equal to $\sigma^n \mc R_1$.

Take any composition series for $M$ and suppose $\sigma^m \mc V$ is the first term of the series. By induction, $\widetilde M \defeq M/\sigma^m \mc V$ is a quotient of $\bigoplus_i \sigma^{n_i} \mc R_{a_i, b_i}$. Since a roof can be covered by a longer roof, $\widetilde M$ is also a quotient of $\bigoplus_i \sigma^{n_i} \mc R_{k_i}$ for $k_i$ sufficiently large. We will see that one can choose the $k_i$ strategically to simplify the proof.

So far, we have the following exact diagram
\begin{equation}
\begin{tikzcd}
{}
    & 
    & 
    & 0 
    & \\
0 \rar
    & \sigma^m \mc V \rar
    & M \rar{\pi}
    & \widetilde M \rar \uar
    & 0 \\
{}
    & 
    & 
    & \bigoplus_i \sigma^{n_i} \mc R_{k_i} \uar{\pi'}
    & 
\end{tikzcd}
\end{equation}
Letting $M'$ be the fiber product of $\pi$ and $\pi'$, the above exact diagram can be extended to
\begin{equation}
\begin{tikzcd}
{}
    & 
    & 0
    & 0 
    & \\
0 \rar
    & \sigma^m \mc V \rar
    & M \rar{\pi} \uar
    & \widetilde M \rar \uar
    & 0 \\
0 \rar
    & \sigma^m \mc V \rar \uar{=}
    & M' \rar \uar
    & \bigoplus_i \sigma^{n_i} \mc R_{k_i} \uar{\pi'} \rar
    & 0
\end{tikzcd}
\end{equation}
The proof will be complete if we can show that $M'$ is a quotient of roofs, so we must analyze the extensions of $\bigoplus_i \sigma^{n_i} \mc R_{k_i}$ by $\sigma^m \mc V$. Since 
\begin{equation}\label{indext}
    \ext^1 \left( \bigoplus_i \sigma^{n_i} \mc R_{k_i}, \sigma^m \mc V \right) \cong \bigoplus_i \ext^1 (\sigma^{n_i} \mc R_{k_i}, \sigma^m \mc V)
\end{equation}
we will see that we can assume, without loss of generality, that $\widetilde M$ is covered by a single roof $\sigma^n \mc R_k$. We split the analysis into cases based on the value of $m$.

If $m=n$, then choose $k$ to be odd. Let $L$ be the length $k-1$ submodule of $\sigma^n \mc R_k$ constituting its ``left half", i.e. $L$ looks like
\begin{equation}
\begin{tikzcd}
\sigma^n \mc V \rar
    & \sigma^{n-1} \mc V \rar
    & \sigma^n \mc V \rar
    & \cdots \rar
    & \sigma^{n-1} \mc V
\end{tikzcd}
\end{equation}
Then $R \defeq \sigma^n \mc R_k / L$ is the length $k$ ``right half" of the roof, which looks like
\begin{equation}
\begin{tikzcd}
\sigma^n \mc V \rar
    & \sigma^{n+1} \mc V \rar
    & \sigma^n \mc V \rar
    & \cdots \rar
    & \sigma^n \mc V
\end{tikzcd}
\end{equation}
Applying $\text{Ext}(-, \sigma^m \mc V)$ to
\begin{equation}
\begin{tikzcd}
0 \rar
    & L \rar
    & \sigma^n \mc R_k \rar
    & R \rar
    & 0
\end{tikzcd}
\end{equation}
produces the long exact sequence
\begin{equation}\label{m=nles}
\begin{tikzcd}[row sep=.1cm]
0 \rar
    & \hom(R, \sigma^m \mc V) \rar
    & \hom(\sigma^n \mc R_k, \sigma^m \mc V) \rar
    & \hom(L, \sigma^m \mc V) \\
{} \rar
    & \ext^1 (R, \sigma^m \mc V) \rar
    & \ext^1 (\sigma^n \mc R_k, \sigma^m \mc V) \rar
    & \ext^1 (L, \sigma^m \mc V)
\end{tikzcd}
\end{equation}
which is
\begin{equation}
\begin{tikzcd}[row sep=.1cm]
0 \rar
    & \mbb C \rar
    & \mbb C \rar
    & 0 \\
{} \rar
    & 0 \rar
    & \ext^1 (\sigma^n \mc R_k, \sigma^m \mc V) \rar
    & 0
\end{tikzcd}
\end{equation}
hence there are no non-trivial extensions. This means that $M' \cong \sigma^m \mc V \oplus \sigma^n \mc R_k$ in this case, which indeed is a (trivial) quotient of roofs.

If $m = n \pm 1$, then choose $k$ to be even so that we have a short exact sequence
\begin{equation}
\begin{tikzcd}
0 \rar
    & \sigma^n \mc R_k \rar
    & \sigma^n \mc W_0^{-,k/2} \oplus \sigma^{n+1} \mc W_0^{+,k/2} \rar
    & \sigma^n \mc V \rar
    & 0
\end{tikzcd}
\end{equation}
Applying $\ext (-, \sigma^m \mc V)$ to this yields
\begin{equation}
\begin{tikzcd}[row sep=.1cm]
0 \rar
    & \hom(\sigma^n \mc V, \sigma^m \mc V) \rar
    & \hom(\sigma^n \mc W_0^{-,k/2} \oplus \sigma^{n+1} \mc W_0^{+,k/n}, \sigma^m \mc V) \rar
    & \hom(\sigma^n \mc R_k, \sigma^m \mc V) \\
{} \rar
    & \ext^1 (\sigma^n \mc V, \sigma^m \mc V) \rar
    & \ext^1 (\sigma^n \mc W_0^{-,k/2} \oplus \sigma^{n+1} \mc W_0^{+,k/n}, \sigma^m \mc V) \rar
    & \ext^1 (\sigma^n \mc R_k, \sigma^m \mc V) \\
{} \rar
    & \ext^2 (\sigma^n \mc V, \sigma^m \mc V)
\end{tikzcd}
\end{equation}
which is
\begin{equation}
\begin{tikzcd}[row sep=.1cm]
0 \rar
    & 0 \rar
    & 0 \rar
    & 0 \\
{} \rar
    & \mbb C \rar
    & \mbb C \rar
    & \ext^1 (\sigma^n \mc R_k, \sigma^m \mc V) \\
{} \rar
    & 0
\end{tikzcd}
\end{equation}
hence there are no non-trivial extensions.

If $m = n + 2$, then choose $k$ to be odd and define $L$ and $R$ as in the $m=n$ case above. Starting from the same short exact sequence, the evaluation of the long exact sequence (\ref{m=nles}) for $m = n + 2$ gives
\begin{equation}
\begin{tikzcd}[row sep=.1cm]
0 \rar
    & 0 \rar
    & 0 \rar
    & 0 \\
{} \rar
    & 0 \rar
    & \ext^1(\sigma^n \mc R_k, \sigma^m \mc V) \rar
    & 0
\end{tikzcd}
\end{equation}
hence there are no non-trivial extensions. The case $m = n - 2$ is similar but one must instead choose $R$ to be the length $k-1$ submodule involving $\sigma^{n+1} \mc V$ and $L$ to be the corresponding length $k$ quotient.

By a similar argument, there are no non-trivial extensions if $|m-n| > 2$.

Therefore, with a suitable choice for $k$, we have shown $\ext^1 (\sigma^{n_i} \mc R_k, \sigma^m \mc V) = 0$. Removing our assumption, the same technique can be used to show that equation (\ref{indext}) is zero (with suitably chosen $k_i$), hence $M'$ is a direct sum of roofs, as desired.
\end{proof}
As a corollary of these two theorems, we now understand why $\cbg$ does not have enough projectives. Suppose we suspect an object $P$ to be projective in $\mc C_{\bg, \lambda}$. By the previous two theorems, $P$ is a quotient of a direct sum of chains $C$ described by a surjective map $\map{f}{C}{P}$. Let $C'$ be the direct sum of the same chains that appear in $C$ but make them, say, 3 times as long. From our remarks earlier in the section, there exists another surjection $\map{\pi}{C'}{C}$ that maps the top third of $C'$ onto $C$ and the bottom two-thirds of $C'$ to 0. If $P$ is projective, then there should exist some map $g$ making the following commute
\begin{equation}
\begin{tikzcd}
{}
    &
    & C' \dar{f \circ \pi} \\
C \rar{f}
    & P \rar{\unit} \urar{g}
    & P
\end{tikzcd}
\end{equation}
Such a map does not exist because the image of $g \circ f$ is contained in the kernel of $\pi$: chains cannot map to composition factors in another chain that are ``higher up the chain" (i.e. further from the bottom of the chain) than the length of the original chain, hence $g \circ f$ maps to the bottom two-thirds of $C'$. Thus $P$ is not projective. The argument for objects in $\mc C_{\bg, \mbb Z}$ is similar; simply replace each occurrence of ``chain" with ``roof". In conclusion, we have actually managed to show that $\cbg$ does not even contain a single projective object!

\subsection{Tensor Structure}\label{2.4}

We briefly recall the relevant definitions of the $P(z)$-fusion product \cite[Definition 4.15]{hlz3} for the benefit of the unfamiliar reader. Given objects $A,B,$ and $C$ of $\cbg$, a \emph{$P(z)$-intertwining map} of type $\binom{C}{A\, B}$ is a linear map from $A \otimes B$ to $\ol C$ satisfying certain compatibility conditions \cite[Definition 4.2]{hlz3}. Here $\ol C$ denotes the completion of $C$ with respect to its $L_0$ grading. For any two modules $A,B$ in $\cbg$, a \emph{$P(z)$-product} of $A$ and $B$ is a module $C$ of $\cbg$ together with a $P(z)$-intertwining map $I$ of type $\binom{C}{A\, B}$. Now, a \emph{$P(z)$-fusion product} of $A$ and $B$, which we denote by $(A \times B, I)$, is a universal such $P(z)$-product in the following sense: for any other $P(z)$-product $(D,I')$ of $A$ and $B$, there exists a unique morphism $f$ from $A\times B$ to $D$ such that $I' = \ol f \circ I$. Here $\ol f$ represents the completion of $f$ with respect to the generalized eigenspace decompositions of $A \times B$ and $D$ under $L_0$.

Our category $\cbg$ possesses the structure of a braided tensor category given by the $P(z)$-fusion product (or fusion product in short). It is not straightforward to prove that $\cbg$ satisfies the assumptions in the work of \cite{hlz1}-\cite{hlz8}, since the modules fail to have bounded-from-below conformal weights. Thus we cannot directly conclude that $P(z)$-fusion products actually define a tensor structure on $\cbg$. Moreover, performing computations in $\cbg$ with this universal definition is very difficult in practice. To circumvent this roadblock, we will use the idea of mirror symmetry to connect $\cbg$ to the category of modules for a simple current extension of the VOA associated to $\widehat \gl$. This approach was successfully executed in \cite{adamovic2019fusion} to determine fusion rules in the subcategory of \emph{weight} modules studied by \cite{rw}. The advantage of following this approach to study the larger category of modules of $\widehat \gl$ that we consider is that the grading restriction is automatically satisfied, hence the machinery of \cite{hlz1}-\cite{hlz8} can be applied. In Section \ref{4}, we demonstrate that $\cbg$ sits in a larger category $\text{Rep}^0(\mc V_{ext})$ which is related to the category of modules of $\widehat \gl$. By the work of \cite{creutzig2020direct, creutzig2020tensor}, we can show that $\text{Rep}^0(\mc V_{ext})$  is a braided tensor category defined by $P(z)$-intertwining maps, which will lead to a braided tensor structure on $\cbg$. Moreover, using our classification results in Section \ref{2.3}, we will prove the 3d mirror symmetry statement, namely the second half of Theorem \ref{LALBequiv}. Let us now turn to $\widehat \gl$, which is the next main ingredient of our story.

\section{The affine Lie superalgebra $\widehat \gl$}\label{3}

In this section, we will study the representation theory of the affine Lie superalgebra $\widehat \gl$. In Section \ref{3.1}, we recall the definition of the Lie superalgebra $\gl$ and its category of finite-dimensional modules. In Section \ref{3.2}, we review the affine Lie superalgebra $\widehat \gl$, following the work of \cite{creutzig2020tensor}. In Section \ref{3.3}, we describe the category $KL$, again following the work of \cite{creutzig2020tensor}; we then proceed to prove in Proposition \ref{glaffinerep} that $KL$ is closely related to the category of finite-dimensional modules of $\gl$. In Section \ref{3.4}, we use the result of \cite{creutzig2020tensor} to compute the fusion product of indecomposable modules of $\widehat \gl$. We especially see how the structure of the representation categories of $\gl$ and $\widehat \gl$ are related to each other.

\subsection{The Lie superalgebra $\gl$}\label{3.1}
The Lie superalgebra $\gl$ is defined as the endomorphism algebra of the superspace $\mathbb{C}^{1|1}$. This Lie algebra has basis
\begin{equation}
N = \frac{1}{2} \left( 
    \begin{array}{cc}
    1 & 0 \\
    0 & -1
    \end{array}\right)
~~~
E = \left(
    \begin{array}{cc}
    1 & 0 \\
    0 & 1
    \end{array}\right)
~~~
\psi^+ = \left(
    \begin{array}{cc}
    0 & 1 \\
    0 & 0
    \end{array}\right)
~~~
\psi^- = \left(
    \begin{array}{cc}
    0 & 0 \\
    1 & 0
    \end{array}\right)
\end{equation}
where $N$ and $E$ are even and $\psi^\pm$ are odd. The non-trivial commutation relations are
\begin{equation}\label{eqfingl11}
[N,\psi^\pm]=\pm \psi^\pm,~~\{\psi^+,\psi^-\}=E.
\end{equation}
There is a supersymmetric, even, non-degenerate, invariant bilinear form $\kappa(\cdot,\cdot)$ on $\gl$ whose non-zero values on basis elements are
\begin{equation}
\kappa(N,E)=\kappa(E,N)=1,~~~\kappa(\psi^+,\psi^-)=-\kappa(\psi^-,\psi^+)=1.
\end{equation}

\subsubsection{The relevant representation category}
To make contact with $\cbg$, we must carefully choose the $\gl$ representation subcategory that we study. For now, let $\mc C$ be the supercategory of finite-dimensional modules for the Lie superalgebra $\gl$, enriched to contain morphisms of odd degree \cite{BrundanSuperStuff}. We do not require $N$ nor $E$ to act semisimply on $\mc C$, unlike most of the literature on the representation theory of $\gl$. However, we will eventually restrict to a full subcategory wherein a specified linear combination of $N$ and $E$ \emph{does} act semi-simply. With this in mind, we organize modules into families parametrized by a $\cp^1$-valued parameter $x$, which indicates that $N-xE$ acts semi-simply (the case $x = \infty$ is understood to mean that $E$ acts semi-simply). We drop the label when the family contains exactly one module.

In the following sections, one may notice that we have described module families and given proofs in separate cases based on the value of $x$. Let us briefly digress to explain why this was necessary. Our work will often utilize a certain linear combination of $N$ and $E$ that possesses a non-zero nilpotent part. If it were possible to provide the same module descriptions and proofs for every $x \in \cp^1$, we would need a continuous parameterization of the linear combinations $N - \alpha(x) E$, with $\alpha(x) \in \cp^1$, that acts non-semisimply on modules with label $x$. This is equivalent to finding a continuous map $\map{\alpha}{\cp^1}{\cp^1}$ such that $\alpha(x) \ne x$ for all $x \in \cp^1$. The Brouwer fixed point theorem tells us this cannot be done.

\subsubsection{Elementary modules}
We introduce some basic objects in $\mc C$ that will be heavily used in this paper.
\begin{enumerate}[leftmargin=0pt]
\item \textbf{Singletons: $A^k_n \quad (k \in \mbb Z,\, n \in \mbb C)$} \\

$A^k_n$ is a $k$-dimensional module on which $E$ and $\psi^\pm$ act as zero, and $N$ has a rank $k$ Jordan block with eigenvalue $n$.

\item \textbf{Typical chains:} $V^k_{n,e,x} \quad (k \in \mbb Z, \, n \in \mbb C,\, e \in \mbb C \setminus \{0\})$ \\

The chain $V^k_{n,e,x}$ for $x \ne \infty$ is uniquely characterized (up to isomorphism) by the following property: there exists a vector $v_1$ such that
\begin{itemize}
\item $v_1 \in \geig (N, n + \frac 1 2)$

\item $\psi^+ v_1 = 0$

\item $(E-e)^k v_1 = 0$ but $(E-e)^{k-1} v_1 \ne 0$ and $(E-e)^{k-1} \psi^- v_1 \ne 0$

\item Defining $v_j = (E-e)^{j-1} v_1$ and $\ol v_j = \psi^- v_j$ for $1 \le j \le k$, $\{ v_1, \ol v_1, \ldots , v_k , \ol v_k \}$ form a basis for $V^k_{n,e,x}$.

\item $N - xE$ acts semi-simply
\end{itemize}
where $\geig (N, \lambda)$ denotes the generalized eigenspace of $N$ corresponding to eigenvalue $\lambda$. One can use the Jordan-Chevalley decomposition to check that $N$ is composed of a Jordan block of rank $k$ corresponding to eigenvalue $n+\frac 1 2$ and a Jordan block of rank $k$ corresponding to eigenvalue $n - \frac 1 2$.

The chain $V^k_{n,e,\infty}$ is uniquely characterized (up to isomorphism) by the following property: there exists a vector $v_1$ such that
\begin{itemize}
\item $\psi^+ v_1 = 0$

\item $(N-(n+ \frac 1 2))^k v_1 = 0$ but $(N-(n+\frac 1 2))^{k-1} v_1 \ne 0$ and $(N-(n- \frac 1 2))^{k-1} \psi^- v_1 \ne 0$

\item Defining $v_j = (N-(n + \frac 1 2))^{j-1} v_1$ and $\ol v_j = \psi^- v_j$ for $1 \le j \le k$, $\{ v_1, \ol v_1, \ldots , v_k , \ol v_k \}$ form a basis for $V^k_{n,e,\infty}$.

\item $E$ acts semi-simply
\end{itemize}

As a visual aid, we depict $V^3_{n,e,x}$ ($x \ne \infty$) here:
\begin{equation}
\begin{tikzcd}[every arrow/.append style={mapsto}, column sep=large, row sep=large]
v_1 \arrow[r, bend left, ] \arrow[d, squiggly, "x"']
	& \ol v_1 \arrow[dotted, l, bend left, "e"'] \arrow[dl, dotted] \arrow[d, squiggly, "x"] \\
v_2 \arrow[r, bend left, start anchor={[yshift=-3pt]}, end anchor={[yshift=-3pt]}] \arrow[d, squiggly, "x"']
	& \ol v_2 \arrow[dotted, l, bend left, "e"'] \arrow[dl, dotted] \arrow[d, squiggly, "x"] \\
v_3 \arrow[r, bend left, start anchor={[yshift=-3pt]}, end anchor={[yshift=-3pt]}]
	& \ol v_3 \arrow[dotted, l, bend left, "e"']
\end{tikzcd}
\end{equation}
The squiggly lines represent the off-diagonal diagonal action of $N$; the semi-simple part of $N$ is $n + \frac 1 2$ on the left column, and $n - \frac 1 2$ on the right column. The dotted arrows represent the action of $\psi^+$ and the solid arrows represent the action of $\psi^-$.

For $x = \infty$, $V^3_{n,e,\infty}$ looks like
\begin{equation}
\begin{tikzcd}[every arrow/.append style={mapsto}, column sep=large, row sep=large]
v_1 \arrow[r, bend left] \arrow[d, squiggly]
	& \ol v_1 \arrow[dotted, l, bend left, "e"'] \arrow[d, squiggly] \\
v_2 \arrow[r, bend left] \arrow[d, squiggly]
	& \ol v_2 \arrow[dotted, l, bend left, "e"'] \arrow[d, squiggly] \\
v_3 \arrow[r, bend left] 
	& \ol v_3 \arrow[dotted, l, bend left, "e"']
\end{tikzcd}
\end{equation}

Note that $V^k_{n,e,x}$ is the unique $(k-1)\th$-iterated self extension of $V_{n,e,x}$.

\item \textbf{Atypical chains $V^k_{n,0,\pm,x} \quad (k \in \mbb Z, \, n \in \mbb C)$} \\

There are two ways to take the heuristic limit $\lim_{e \to 0} V^k_{n,e,x}$, and each results in a distinct module. They look very similar to the typical chains, but we instead choose to describe them in terms of their Loewy diagram. For $x \ne \infty$, the \emph{positive} chains look like
\begin{equation}
V^k_{n,0,+,x} \defeq 
\begin{tikzcd}
A^1_{n - \frac 1 2} \rar
    & A^1_{n + \frac 1 2} \rar
    & \cdots \rar
    & A^1_{n - \frac 1 2} \rar
    & A^1_{n + \frac 1 2}
\end{tikzcd}
\end{equation}
whereas the \emph{negative} chains look like
\begin{equation}
V^k_{n,0,-,x} \defeq 
\begin{tikzcd}
A^1_{n + \frac 1 2} \rar
    & A^1_{n - \frac 1 2} \rar
    & \cdots \rar
    & A^1_{n + \frac 1 2} \rar
    & A^1_{n - \frac 1 2}
\end{tikzcd}
\end{equation}
Together with the semisimplicity condition on $N-xE$, these Loewy diagrams uniquely describe the atypical chains. The Loewy diagrams look slightly different when $x = \infty$, but this case will not be important for us.

\item \textbf{Diamonds $P^k_{n,x} \quad (k \in \mbb Z, n \in \mbb C)$} \\

The diamond $P^k_{n,x}$ for $x \ne \infty$ is uniquely characterized (up to isomorphism) by the following property: there exists a vector $v_1$ such that
\begin{itemize}
    \item $v_1 \in \geig(N,n)$
    \item $E^k v_1 = 0$
    \item defining
        \begin{equation}
            w_1 = \psi^+ v_1 \qquad
        v_j = \begin{cases}
            \psi^- v_{j-1}
                & j \text{ even} \\
            \psi^+ v_{j-1}
                & j \text{ odd}
            \end{cases} \qquad
        w_j = \begin{cases}
            \psi^- w_{j-1}
                & j \text{ even} \\
            \psi^+ w_{j-1}
                & j \text{ odd}
            \end{cases}
        \end{equation}
        for $2 \le j \le 2k$, we have that $\{ v_1, w_1, \ldots, v_{2k}, w_{2k} \}$ forms a basis for $P^k_{n,x}$
    \item $N-xE$ acts semi-simply
\end{itemize}

The diamond $P^k_{n,\infty}$ is uniquely characterized (up to isomorphism) by the following property: there exists a vector $v_1$ such that
\begin{itemize}
    \item $v_1 \in \geig(N,n)$
    \item defining
        \begin{equation}
        v_j = (N-n)^{j-1} v_1 \qquad
        x_j = \psi^+ v_j \qquad
        y_j = \psi^- v_j \qquad
        w_j = \psi^- \psi^+v_j
        \end{equation}
        for $1 \le j \le k$, we have that $\{ v_1, x_1, y_1, w_1, \ldots, v_k, x_k, y_k, w_k \}$ forms a basis for $P^k_{n,\infty}$
    \item $E = 0$ on the entire module
\end{itemize}

We first depict the diamonds when $x=0$ to remove some clutter that might otherwise obfuscate their core structure. $P_{n,0}$ looks like
\begin{equation}
\begin{tikzcd}[every arrow/.append style={mapsto}, cramped]
{}
    & v_1 \ar[rd, "\psi^-"] \ar[ld, "\psi^+",']
    & \\
w_1 \ar[rd, "\psi^-",']
    & 
    & v_2 \ar[ld, "-\psi^+"] \\
{}
    & w_2
\end{tikzcd}
\end{equation}
where $N v_1 = n v_1$. There exist iterated self-extensions of $P_{n,0}$ called $P_{n,0}^k$. The positive integer $k$ refers to the ``number of $P_{n,0}$'s it contains''. For example, $P_{n,0}^2$ is
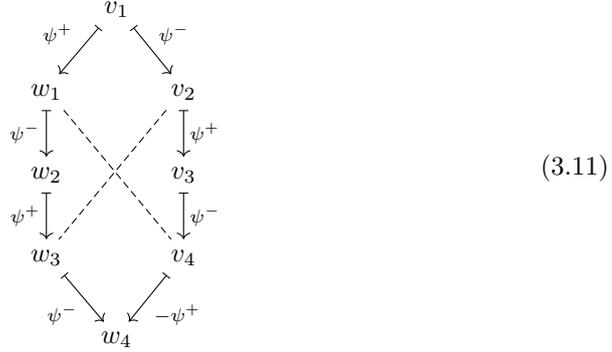
\begin{equation}
\begin{tikzcd}[every arrow/.append style={mapsto}, column sep=tiny]
 & v_1 \dlar[']{\psi^+} \ar[rd,"\psi^-"] &
\\
w_1 \dar[']{\psi^-} \ar[ddrr, dash, dashed]
	& 
	& v_2 \dar{\psi^+} \ar[ddll, dash, dashed]
\\
w_2 \dar[']{\psi^+}
	& 
	& v_3 \dar{\psi^-}
\\
w_3 \drar[']{\psi^-}
	& 
	& v_4 \dlar{-\psi^+}
\\
 & w_4
\end{tikzcd}
\end{equation}
The dashed lines here illustrate the $P_{n,0}$'s that $P^2_{n,0}$ contains as submodules and quotients. The submodule (bottom diamond) is generated by $w_2 + v_3$. When we quotient by this bottom diamond, the bottom vector of the quotient (top diamond) is given by the equivalence class of $w_2 - v_3$.

Restoring $x$, we may illustrate $P^3_{n,x}$ with $x \ne \infty$ as:
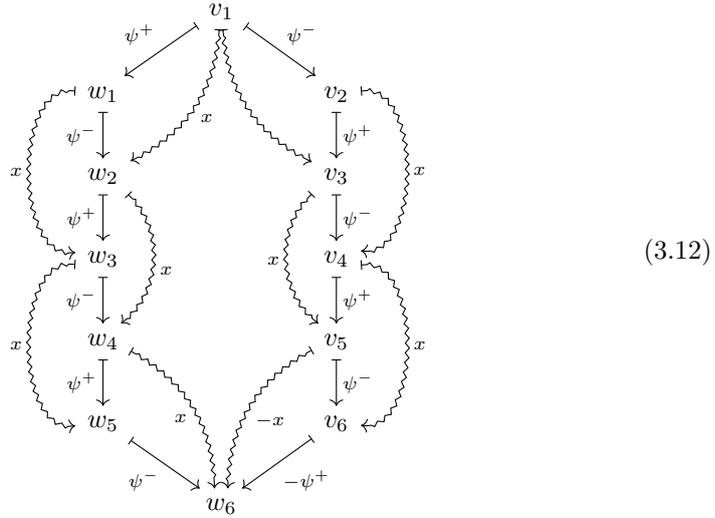
\begin{equation}
\begin{tikzcd}[every arrow/.append style={mapsto}]
{}
    & v_1 \dlar[']{\psi^+} \drar{\psi^-} \arrow[ddl, bend left, squiggly, "x"] \arrow[ddr, squiggly, bend right]
    & \\
w_1 \dar[']{\psi^-} \arrow[dd, bend right=100, squiggly, "x"', start anchor=west, end anchor={[yshift=5pt]}]
	& 
	& v_2 \dar{\psi^+} \arrow[dd, bend left=100, squiggly, "x", start anchor=east, end anchor={[yshift=5pt]}] \\
w_2 \dar[']{\psi^+} \arrow[dd, bend left=50, squiggly, "x"]
	& 
	& v_3 \dar{\psi^-} \arrow[dd, bend right=50, squiggly, "x"'] \\
w_3 \dar[']{\psi^-} \arrow[dd, bend right=100, squiggly, "x"', start anchor={[yshift=-5pt]}]
	& 
	& v_4 \dar{\psi^+} \arrow[dd, bend left=100, squiggly, "x", start anchor={[yshift=-5pt]}] \\
w_4 \dar[']{\psi^+} \arrow[ddr, bend left, squiggly, "x"', end anchor={[xshift=-2pt]}]
	& 
	& v_5 \dar{\psi^-} \arrow[ddl, bend right, squiggly, "-x", end anchor={[xshift=2pt]}] \\
w_5 \drar[']{\psi^-}
	& {}
	& v_6 \dlar{-\psi^+} \\
{}
    & w_6
    & 
\end{tikzcd}
\end{equation}
where the squiggly arrows represent the off-diagonal action of $N$. 

The module $P^2_{n,\infty}$ is drawn as:
\begin{equation}
\begin{tikzcd}[every arrow/.append style={mapsto}]
{}
    & v_1 \dlar[']{\psi^+} \drar{\psi^-} \arrow[ddd, squiggly, bend right]
    & \\
x_1 \drar[']{\psi^-} \arrow[ddd, squiggly]
	& 
	& y_1 \dlar{-\psi^+} \arrow[ddd, squiggly]\\
{}
    & w_1 \arrow[ddd, squiggly, bend left]
    & \\
{}
    & v_2 \dlar[']{\psi^+} \drar{\psi^-}
    & \\
x_2 \drar[']{\psi^-}
	& 
	& y_2 \dlar{-\psi^+} \\
{}
    & w_2
    & 
\end{tikzcd}
\end{equation}
where the squiggly arrows represent the off-diagonal action of $N$.
\end{enumerate}

At this point, the reader may wish to look back at the $\bg$ modules (their Loewy diagrams, in particular) introduced in Section \ref{bgrepcat} to get a sense of what the categorical equivalence will ultimately look like. A small detail that we have swept under the rug is the $\mbb Z_2$ grading on the modules and morphisms in the supercategory $\mc C$. Given an object $X \in \text{Ob}(\mc C)$, its parity conjugate $\Pi X$ is also an object in $\mc C$. There is no corresponding notion of the parity-shifted version of a module in $\cbg$, therefore one might be concerned about the correctness of the equivalence. It turns out that $X$ and $\Pi X$ are isomorphic in $\mc C$, albeit via an \emph{odd} isomorphism, for which there is no analogue in $\cbg$. To foreshadow the resolution, we will instead find that our category matches $\cbg \boxtimes \text{SVect}$. This does not impede our ultimate goal because SVect has a nearly transparent effect on the tensor structure, so we can straightforwardly extract the fusion structure on $\cbg$.

\subsubsection{Tensor structure}
As we have mentioned before, we still need to pass through a few more categories and constructions before connecting with $\cbg$. However, much of the tensor structure on $\cbg$ can ultimately be obtained from the structure on $\mc C$, where computations are much easier. In this section, we give the tensor product decompositions that we computed that will be relevant for the remainder of this paper.

The action of $\gl$ on a tensor product of super modules is
\begin{equation}
x \cdot (v\otimes w) = (x \cdot v) \otimes w + (-1)^{|x||v|}v \otimes (x \cdot w)
\end{equation}
for homogeneous $x$ and $v$. Also the map $\map{\tau}{V \otimes W}{W \otimes V}$ given by $\tau(x \otimes y) = (-1)^{|x| |y|} y \otimes x$ is an isomorphism. Since we will eventually restrict to a subcategory labeled by a particular value of $x \in \cp^1$, we only compute tensor products between modules with the same $x$ label and will ignore the $x = \infty$ case. By tedious but not conceptually challenging computations, one sees
\begin{align}
V^s_{n, e, x} \otimes V^t_{m, f, x}
    & \cong \begin{cases}
    \bigoplus_{k = 0}^{\min(s,t) - 1} \left ( V^{s+t-1-2k}_{n+m+\frac 1 2, e+f, x} \oplus V^{s+t-1-2k}_{n+m-\frac 1 2, e+f, x} \right )
        & e+f \ne 0 \\
    \bigoplus_{k=0}^{\min(s,t)-1} P^{s+t-1-2k}_{n+m,x}
        & e+f=0
    \end{cases} \\
V^s_{n,e,x} \otimes V^t_{m,0,\pm,x}
    & \cong \bigoplus_{k=0}^{\min(s,t)-1} \left( V^{s+t-1-2k}_{n+m+\frac 1 2, e, x} \oplus V^{s+t-1-2k}_{n+m-\frac 1 2, e, x} \right ) \\
V^s_{n,0,\epsilon_1,x} \otimes V^t_{m,0,\epsilon_2,x}
    & \cong \begin{cases}
        V^t_{n+m+\frac 1 2,0,\epsilon_1,x} \oplus V^t_{n+m-\frac 1 2,0,\epsilon_2,x}
            & \epsilon_1 = \epsilon_2, \, s=1 \\
        \bigoplus_{k=0}^{\min(s,t)-1} P^{s+t-1-2k}_{n+m,x}
            & \epsilon_1 = - \epsilon_2
        \end{cases}
\end{align}

The tensor products where $\epsilon_1 = \epsilon_2$ with $s>1$ can generate new indecomposables that we will not write down here. Furthermore:
\begin{align}
P^s_{n,x} \otimes P^t_{m,x}
    & \cong \bigoplus_{j =0}^{\min(s,t) - 1} \left ( P^{s+t - 1 - 2j}_{m+n+1, x} \oplus 2\, P^{s+t - 1 - 2j}_{m+n, x} \oplus P^{s+t - 1 - 2j}_{m+n-1, x} \right ) \\
P^s_{n,x} \otimes V^t_{m,e,x}
    & \cong \bigoplus_{j=0}^{\min(s,t)-1} \left( V^{s+t-1-2j}_{m+n+1,e,x} \oplus 2\, V^{s+t-1-2j}_{m+n,e,x} \oplus V^{s+t-1-2j}_{m+n-1,e,x} \right )
\end{align}

\subsection{The vertex superalgebra $\widehat \gl$}\label{3.2}

Now we turn to the definition of the affine Lie superalgebra $\widehat \gl$ associated to the bilinear form $\kappa$. It is defined as the super vector space $\gl\otimes \mathbb{C}[t,t^{-1}]\oplus \mathbb{C}\mathbf{k}$ where $\mathbb{C}[t,t^{-1}]$ and $\mathbf{k}$ are even, together with the following non-trivial Lie brackets:
\begin{equation}
[N_r,E_s]=r\mathbf{k}\delta_{r+s,0},~~[N_r,\psi^\pm_s]=\pm \psi^\pm_{r+s},~~\{\psi_r^+,\psi_s^-\}=E_{r+s}+r\mathbf{k}\delta_{r+s,0}
\end{equation}
where $a_r$ denotes $a\otimes t^r$. 

Given a module $M$ of $\gl$, one may obtain a module of $\widehat \gl$ as follows: one first view $M$ as a module of $\gl\otimes \mathbb{C}[t]\oplus \mathbb{C}\mathbf{k}$ such that $\gl\otimes t\mathbb{C}[t]$ acts trivially and $\mathbf{k}$ acts as a number $k\in \mathbb{C}$; one may then define the induced module:
\begin{equation}
    \widehat{M}^k=\mathcal{U}(\widehat \gl)\otimes_{\mathcal{U}(\gl\otimes \mathbb{C}[t]\oplus \mathbb{C}\mathbf{k})} M
\end{equation}
as a representation of $\widehat \gl$. When $M$ is the trivial module, $\widehat{M}^k$ has the structure of a vertex operator superalgebra (VOSA), which we denote by $V_k(\widehat \gl)$. For general $M$, $\widehat{M}^k$ becomes a module of $V_k(\widehat \gl)$. The assignment $M\to \widehat{M}^k$ defines a functor ${\mc I}nd$, which we will call the induction functor. 

\vspace{5pt}

\noindent\textbf{Remark.} It turns out, as explained in \cite{creutzig2020tensor}, that for different choice of $k\ne 0$, the vertex algebras $V_k(\widehat \gl)$ are isomorphic to each other. Thus we will once and for all fix $k=1$, and drop $k$ from all notations. 

\vspace{5pt}
The VOSA $V(\widehat \gl)$ has the following conformal element:
\begin{equation}
    \omega=\frac{1}{2}(N_{-1}E_{-1}+E_{-1}N_{-1}-\psi_{-1}^+\psi_{-1}^-+\psi_{-1}^-\psi_{-1}^+)+\frac{1}{2}E_{-1}^2
\end{equation}
with the associated Virasoro zero mode:
\begin{equation}
\begin{aligned}
    L_0=&\sum_{r>0}\left(N_{-r}E_{r}+E_{-r}N_{r}-\psi_{-r}^+\psi_{r}^-+\psi_{-r}^-\psi_{-r}^+\right)+\sum_{r>0}E_{-r}E_r\\& + (N_0+E_0/2)E_0-\frac{1}{2}(\psi_0^+\psi_0^--\psi_0^-\psi_0^+)
    \end{aligned}
\end{equation}
It also enjoys spectral flow symmetries $\sigma^l$:
\begin{equation}
\sigma^l(N_r)=N_r, ~~\sigma^l(E_r)=E_r-l\delta_{r,0},~~\sigma(\psi_r^\pm)=\psi^\pm_{r\mp l}
\end{equation}
as well as the conjugation $w$:
\begin{equation}
w(N_r)=-N_r,~~w(E_r)=-E_r,~~w(\psi_r^+)=\psi_r^-,~~w(\psi_r^-)=-\psi_r^+.
\end{equation}
These can be used to twist representations to obtain new representations.

\subsection{The Kazhdan-Lusztig category}\label{3.3}

In this section, we recall the Kazhdan-Lusztig category $KL$ of representations of $V(\widehat \gl)$ that we are interested in. This category is characterized by satisfying certain weight constraint. For a generalized $V(\widehat \gl)$ module $W$, it is called \emph{finite-length} if it has a finite composition series of irreducible $V(\widehat \gl)$ modules. $W$ is called \emph{grading restricted} if it is graded by generalized conformal weights (the generalized eigenvalues of $L_0$) and the generalized conformal weights are bounded from below. For more details, see \cite{creutzig2017tensor}. 

\newtheorem{DefKL}{Definition}[section]

\begin{DefKL}
The Kazhdan-Lusztig category $KL$ is defined as the supercategory of finite-length grading-restricted generalized $V(\widehat \gl)$ modules. 
\end{DefKL}

Any simple module of this category is generated by its lowest conformal weight space, which is a finite-dimensional representation of $\gl$, and thus any simple module is a quotient of $\widehat{V}_{n,e}$ for $e\ne 0$ or $\widehat{A}_{n}$. In what follows, when we write $\widehat{M}$, we always mean the image of a module $M$ of $\gl$ under the induction functor ${\mc I}nd$.
The following is shown in \cite{creutzig2013relating}:
 
 \begin{itemize}

\item $\widehat{V}_{n,e}$ is irreducible iff $e\notin\mathbb{Z}$.

\item When $e=0$, $\widehat{A}_{n}$ is irreducible, and there are non-split exact sequences:
\begin{equation}
\begin{tikzcd}
 0\rar & \widehat{A}_{n-\frac{1}{2}}\rar & \widehat{V}_{n,0,+}
\rar &\widehat{A}_{n+\frac{1}{2}}\rar &0\\
0\rar &\widehat{A}_{n+\frac{1}{2}}\rar &\widehat{V}_{n,0,-}
\rar & \widehat{A}_{n-\frac{1}{2}}\rar & 0
\end{tikzcd}
\end{equation}

\item When $e\in\mathbb{Z}\setminus \{0\}$, there is additional simple modules $\widehat{A}_{n,e}$. They fit in the following short exact sequences:
\begin{equation}
\begin{tikzcd}
0\rar &  \widehat{A}_{n+1,e}\rar & \widehat{V}_{n,e}
\rar & \widehat{A}_{n,e}\rar & 0 & (e>0)\\
 0\rar & \widehat{A}_{n-1,e}\rar & \widehat{V}_{n,e}
\rar & \widehat{A}_{n,e}\rar & 0 & (e<0)
\end{tikzcd}
\end{equation}
\end{itemize}

\noindent\textbf{Remark.} To unify the notation, we will write $\widehat{A}_{n,0}$ for $\widehat{A}_n$. The modules $\widehat{A}_{n,e}$ for $e\in \mathbb{Z}$ are called simple currents. 

\vspace{5pt}

Since $E_0$ is a central element, any representation of $V(\widehat \gl)$ can be decomposed into direct sums according to the generalized eigenvalues of $E_0$, which is possible by finite-length property. We may thus write 
\begin{equation}
KL=\bigoplus_{e\in \mathbb{C}}KL_e    
\end{equation}
where $KL_e$ is the subcategory such that the generalized eigenvalue of $E_0$ is $e$. From the description of the above simple modules, it is clear that $KL_e$ for $e\notin\mathbb{Z}$ is generated by $\widehat{V}_{n,e}$ and for $m\in \mathbb{Z}$, $KL_m$ is generated by $\widehat{A}_{n,m}$. There is a similar decomposition of $\mc C$, the category of finite-dimensional representations of $\gl$, so we write ${\mc C}_e$ to be the subcategory of $\mc C$ where the action of $E$ has generalized weight $e$. Clearly, induction is a functor from ${\mc C}_e$ to $KL_e$. In fact, more is true about this induction:

\newtheorem{glaffinerep}[DefKL]{Proposition}

\begin{glaffinerep}\label{glaffinerep}
When $e\notin\mathbb{Z}$ or $e=0$, induction functor ${\mc I}nd$ gives an equivalence of abelian supercategories:
\begin{equation}\label{eqCeKLe}
\mathcal{C}_e\cong KL_e
\end{equation} 
\end{glaffinerep}

\begin{proof}
We only consider the case $e\notin \mathbb{Z}$, since the proof for $e=0$ is almost identical. By definition, any simple module in $KL_e$ comes from induction. We now show that induction is essentially surjective. For this, choose $M\in KL_e$, we will use induction on the length of $M$. Choose a maximal sub-module $N\subseteq M$. Then by assumption $N$ is induced from $\mathcal{C}_e$ and $M/N\cong \widehat{V}_{n,e}$ for some $n$. We claim that we may choose generators of $M$ that are annihilated by all the positive modes of $\widehat \gl$. We may choose such generators for $N$ and $M/N$, say $n_1,\ldots,n_k$ and $\overline{m}$, that are generates the lowest-weight spaces for the respective module. Choosing a pre-image $m$ of $\overline{m}$, we will adjust $m$ step by step so that it is annihilated by all positive modes. 

First, suppose there exists $t>0$ such that $E_tm\ne 0$. By grading restriction, $E_t^lm=0$ for some $l>1$. Using $[E_{t}^l,N_{-t}]=(tl)E_t^{l-1}$, we see that $E_t^lN_{-t}m=(tl)E_t^{l-1}m$. Let $m'=m-\frac{1}{tl}E_tN_{-t}m$, then $E_t^{l-1}m'=0$, and $\overline{m'}=\overline{m}-\frac{1}{l}\overline{m}\ne 0$ since $l>1$. This is again a generator. Using this procedure, we may adjust $m$ such that $E_tm=0$ for all $t>0$.

Suppose $\psi_t^+m\ne 0$ for some $t>0$, and assume $(\psi_t^+)^lm=0$ for some $l>1$. Since $\{(\psi_t^+)^l,\psi_{-t}^-\}=l(E_0+t)(\psi_t^+)^{l-1}$, we have $(\psi_t^+)^l\psi_{-t}^-m=l(\psi_t^+)^{l-1}(E_0+t)m$. Let $m'=l(E_0+t)m-\psi_{t}^+\psi_{-t}^-m$, then $(\psi_t^+)^{l-1}m'=0$, and $\overline{m'}=(l-1)(e+t)\overline{m}$ which is still a generator since $e+t\ne 0$ and $l-1\ne 0$. If the starting $m$ satisfies $\psi_{j}^\pm m=E_km=0$ for $j>t$ and $k>0$, then after the adjustment, $\psi_{j}^\pm m'=E_km'=0$ for $j\geq t$ and $k>0$. We may then adjust $m$ downward from the largest $t$ such that $\psi_t^\pm m\ne 0$ , and obtain $m$ that is annihilated by all $\psi_t^\pm, E_t$ for $t>0$.

Finally we perform a similar procedure for $N_t$ to obtain the final $m$. We may also assume that $N_0\overline{m}=(n-1/2)\overline{m}$ for some $n\in \mathbb{C}$, which means that $(N_0-n+1/2)m\in N$. This is an element annihilated by all positive modes of $\widehat \gl$, which then must be in the lowest conformal weight space of $N$. 

Now consider the $\gl$ module $V$ generated by $m$ and the $n_i$. By the above consideration, this module is finite dimensional. The map $V\to M$ then induces a surjection $\widehat{V}\to M$ as $V(\widehat \gl)$ modules. We claim that the kernel $K$ must be of the form $\widehat{W}$ for some $W\subseteq V$. To prove the claim, we use an inductive argument on the number of composition factors of $\widehat{V}$. Choose a minimal sub-module $V_{n',e}$ of $V$ and let $U$ be the quotient $V/V_{n',e}$. We have a short exact sequence:
\begin{equation}
\begin{tikzcd}
    0 \rar & \widehat{V}_{n',e}\rar &  \widehat{V}\rar & \widehat{U}\rar & 0
    \end{tikzcd}
\end{equation}
If $K\cap \widehat{V}_{n',e}=0$, then we are done by inductive hypothesis. Otherwise, $K$ fits in the exact sequence:
\begin{equation}
\begin{tikzcd}
    0 \rar
        & \widehat{V}_{n',e} \rar 
        & \widehat{V} \rar
        & \widehat{U} \rar
        & 0 \\
    0 \rar
        & \widehat{V}_{n',e} \rar \arrow[u, "="]
        & K \rar \uar
        & K/\widehat{V}_{n',e} \rar \uar
        & 0
\end{tikzcd}
\end{equation}
By induction, $K/\widehat{V}_{n',e}$ is equal to $\widehat{L}$ for some $L\subseteq U$. Then by counting composition factors, it is clear that $K\cong \widehat{W}$ where $W$ is the pre-image of $L$ under the map $V\to U$. Thus $M$ is of the form $\widehat{V/W}$ as desired. 
\end{proof}

Because of this, the difficulty in the study of $KL$ lies in understanding $KL_n$ for $ n\in \mathbb{Z}\setminus \{0\}$. We will see in the next section, that fusion product with simple currents gives a way to study them. The objects in $KL_e$ for $e\in \mathbb{Z}$ are called atypical modules, while those in $KL_e$ for $e\notin \mathbb{Z}$ are called typical modules(c.f. Section \ref{2.3}).

\subsection{Fusion structure}\label{3.4}

By the work of \cite{hlz1}-\cite{hlz8}, $P(z)$-intertwining operators define a monoidal structure on $KL$. In this section, we will compute the fusion products of modules by relating the fusion product to the tensor product of $\gl$ modules. 

To start, we have the following statement, which is proved in \cite{creutzig2020tensor}:

\newtheorem{KLrigid}[DefKL]{Theorem}

\begin{KLrigid}
$KL$ is a rigid braided tensor supercategory; moreover, it is a ribbon category with even natural twist $\theta=e^{2\pi i L_0}$.

\end{KLrigid}

The fusion product (if it exists) of two generalized modules $W_1$ and $W_2$ of a VOSA $V$ will be denoted by $W_1\times_{V} W_2$, or simply $W_1\times W_2$ if it is clear what the VOSA is. It turns out that one can understand the fusion structure of $KL$ using the tensor structure of $\gl$ modules. Let $M_1,M_2$ and $M_3$ be finite-dimensional $\gl$ modules. Given an intertwiner operator $\mathcal{Y}$ of type $\binom{\widehat{M}_3}{\widehat{M}_1\widehat{M}_2}$, consider the assignment $\map{\pi(\mathcal{Y})}{M_1\otimes M_2}{M_3}$ given by:
\begin{equation}
\pi(\mathcal{Y})(m_1\otimes m_2)=\pi_0(\mathcal{Y}(m_1,1)m_2)
\end{equation} 
where $\pi_0$ denotes projection onto the lowest conformal weight space. This $\pi(\mathcal{Y})$ is in fact a homomorphism of $\gl$ modules. The inverse of this is established in \cite{creutzig2020tensor}. To state it, let's recall the definition of  taking contragredient dual. For a full definition, see \cite{creutzig2017tensor}. Let $M$ be a finite-dimensional module of $\gl$, then the linear dual $M^*=\mathrm{Hom}(M,\mathbb{C})$ has the structure of a $\gl$ module, induced by the action of $\gl$ on $M$. There is a similar operation on modules of $\widehat \gl$. Given a grading-restricted generalized $\widehat \gl$ module $W$. Write the decomposition of $W$ into generalized conformal weight spaces as:
\begin{equation}
    W=\bigoplus_{h\in \mathbb{C}}W_{[h]}.
\end{equation}
Define $W'$ to be the super vector space:
\begin{equation}
    W'=\bigoplus_{h\in \mathbb{C}}W_{[h]}^*,
\end{equation}
together with the action of $V(\widehat \gl)$ by:
\begin{equation}
    \langle Y_{W'}(v,z)w', w\rangle=\langle w', Y_W(e^{zL_1}(-z^2)^{L_0}v,z^{-1})w\rangle.
\end{equation}
Here $\langle-,-\rangle$ is the natural pairing between $W$ and $W'$, and the above is well-defined by the grading restriction condition. The module $W'$ is called the contragredient dual of $W$. The following is proved in \cite{creutzig2020tensor}:

\newtheorem{glaffine}[DefKL]{Proposition}

\begin{glaffine}\label{glaffine}
Let $M_1$, $M_2$ and $M_3$ be finite-dimensional $\gl$ modules, and let $\map{f}{M_1\otimes M_2}{M_3^*}$ be a homomorphism of $\gl$ modules. Then there exists a unique intertwiner operator $\mathcal{Y}$ of type $\binom{\widehat{M}_3'}{\widehat{M}_1\widehat{M}_2}$  such that $\pi(\mathcal{Y})=f$. 
\end{glaffine}

This Proposition, together with rigidity and Proposition \ref{glaffinerep}, can help us compute fusion product explicitly for any pair modules in $KL$. Let $S_b:=\mathbb{Z}\setminus \{0\}$ and $S_g:=\mathbb{C}\setminus S_b$.

\newtheorem{restricteddual}[DefKL]{Lemma}

\begin{restricteddual}\label{restricteddual}
$\widehat{M_e^*}\cong \widehat{M_e'}$ for any $e\in S_g$. 

\end{restricteddual}

\begin{proof}
By definition, the lowest conformal weight space of $\widehat{M_e'}$ is $M_e^*$, thus the identity map $M_e^*\to M_e^*$ induces a map of VOSA modules $\widehat{M_e^*}\to \widehat{M_e'}$. We claim that this map is injective. Suppose otherwise, then its kernel would contain a minimal module $\widehat{V}_{n-1/2,e}$ for some $n$ in the case $e\ne 0$ or $\widehat{A}_{n,0}$ for the case $e=0$. In any case we denote such minimal module by $\widehat{V}$. Since $\mathrm{Hom}(V, M_e^*)\cong \mathrm{Hom}(\widehat{V}, \widehat{M_e^*})$, the embedding $\widehat{V}\to \widehat{M_e^*}$ comes from the induces map of a injection $V\to M_e^*$, and the map $\widehat{V}\to \widehat{M_e'}$ then comes from the map $\widehat{V}\to \widehat{V^*}'\to \widehat{M_e'}$. The first map is an isomorphism by \cite{creutzig2020tensor}, and the second is nonzero and thus is an embedding, a contradiction. Since this map is injective, by counting the number of composition factors, this is also surjective, so it is an isomorphism.

\end{proof}

As a corollary, we obtain the following:

\newtheorem{goodfus}[DefKL]{Corollary}

\begin{goodfus}\label{goodfus}
In the case when $e,e'\in S_g$ and $e+e'\in S_g$, $\widehat{M_e\otimes M_{e'}}\cong \widehat{M}_e\times \widehat{M}_{e'}$.

\end{goodfus}

\begin{proof}
By Proposition \ref{glaffine} and Lemma \ref{restricteddual}, one has a map $\widehat{M}_e\times \widehat{M}_{e'}\to \widehat{M_e\otimes M_{e'}}$. This is surjective since $ \widehat{M_e\otimes M_{e'}}$ is generated by $M_e\otimes M_{e'}$. Indeed, by Proposition \ref{glaffinerep}, if the image of this intertwiner was contained in some sub-module, it would be a sub-module of the form $\widehat{W}$ for some $W\subseteq M_e\otimes M_{e'}$. This is contradicting the fact that $\pi(\mathcal Y)$ is an isomorphism onto $M_e\otimes M_{e'}$. To conclude the proof, we only need to compute the number of composition factors on both sides. Rigidity of $KL$ implies that fusion is exact, and so the number of composition series can be computed using fusion rules of irreducible modules as in  \cite[Theorem 3.2.4]{creutzig2020tensor}. Comparing this with the composition series of $M_e\otimes M_{e'}$, we conclude that the map is an isomorphism. 

\end{proof}

We can use Corollary \ref{goodfus} to also deal with fusion rules of $M_e$ for $e\in S_b$. Indeed, since fusion product with simple currents preserves socle series \cite{creutzig2013w}, the functor:
\begin{equation}
  \map{\widehat{A}_{n,l}\times -}{KL_e}{KL_{e+l}}
\end{equation}
 is an equivalence, with inverse given by $\widehat{A}_{-n,-l}$. In particular, for any $V \in KL_e$ where $e \in S_b$, we can ``shift" V to be in $KL_0$ using simple currents
\begin{equation}
V =  \widehat{A}_{n,e}\times (\widehat{A}_{-n,-e}\times V).
\end{equation}
Suppose we want to compute $V\times W$ for $V\in KL_e$ and $W\in KL_f$. If either $e,f$ or $e+f\in S_b$, then we can use simple currents to pull the computation into $KL_0$, so that we can apply Corollary \ref{goodfus}. The question now becomes identifying $\widehat{A}_{n,e}\times V$ for various $V$ and $e\in S_b$. It suffices to take $V$ to be indecomposable. We introduce a collection of indecomposable modules that will be relevant for connecting to $\mc C_\bg$ (they either come from induction, or a spectral flow of a module from induction). Note that this is not a complete list of indecomposables in $KL$. 
 \begin{itemize}
 
\item  $\widehat{A}_{n,l}$ for $l\in\mathbb{Z}$. When $l<0$, $\widehat{A}_{n,l}\cong \sigma^l(\widehat{A}_{n+l+\frac{1}{2}, 0})$. 	When $l>0$, $\widehat{A}_{n,l}\cong \sigma^l \circ w(\widehat{A}_{-n-l+\frac{1}{2}, 0})\cong \sigma^l( \widehat{A}_{n+l-\frac{1}{2}, 0})$. 

\item  $\widehat{V}_{n,0,\pm, x}^t$ for $x\in \mathbb{C}$ and $t\in \mathbb{N}$ as well as $\widehat{V}_{n,e, x}^t $ for $e\notin \mathbb{Z}$. These are the induced module from $V_{n,0,\pm, x}^t$ and $V_{n,e, x}^t$ respectively. For $l\in\mathbb{Z}$, define $\widehat{V}_{n,l,\pm,x}^t$ to be $\sigma^l(\widehat{V}_{n+l,0,\pm, x}^t)$. One sees then that when $l<0$, $\widehat{V}_{n,l,-,x}^t=\widehat{V}_{n,l, x}^t$, and when $l>0$, $\widehat{V}_{n,l,+,x}^t=\widehat{V}_{n,l, x}^t$. 

\item $\widehat{P}_{n,x}^t$, which is induced from $P_{n,x}^t$. Define $\widehat{P}_{n-l-\frac{1}{2},l,x}^t$ for $l<0$ to be $\sigma^l(\widehat{P}_{n,x}^t)$ and for $l>0$, define $\widehat{P}_{-n-l+\frac{1}{2},l,x}^t$ to be $\sigma^l\circ w (\widehat{P}_{n,x}^t)$. Note that all these can be obtained by fusion product with simple currents.

\end{itemize}

We will also introduce a function $\epsilon(l)$ on $\mathbb{Z}$ given by:
\begin{equation}
\epsilon(l) = \left\{
\begin{array}{rl}
-\frac{1}{2} & \text{if } l < 0,\\
0 & \text{if } l = 0,\\
\frac{1}{2} & \text{if } l > 0.
\end{array} \right.
\end{equation}
Define $\epsilon(l,l')=\epsilon(l)+\epsilon(l')-\epsilon(l+l')$. 
 
\newtheorem{fuswithsc}[DefKL]{Lemma}

\begin{fuswithsc}\label{fuswithsc}
Let $l'\in S_g$ and $l\in \mathbb{Z}$. One has the following fusion rule with simple currents:
\begin{equation}
\widehat{A}_{n,l}\times \widehat{V}_{n',l',x}^ t= \widehat{V}_{n'+n-\epsilon(l),l'+l,x}^t
\end{equation}

\end{fuswithsc}

\begin{proof}
One has a homomorphism of $V(\widehat \gl)$ modules: 
\begin{equation}
   \map{\mathcal{Y}}{\widehat{V}_{n,l}\times  \widehat{V}_{n',l',x}^t}{\widehat{V_{n,l}\otimes V_{n',l',x}^t}}, 
\end{equation}
 such that $\pi(\mathcal{Y})$ is an isomorphism. Since $l+l'\in S_g$, submodules of $\widehat{V_{n,l}\otimes V_{n',l',x}^t}$ are all induced from submodules of $V_{n,l}\otimes V_{n',l',x}^t$. By construction of $\mathcal{Y}$, if it were not surjective, then $\pi(\mathcal{Y})$ wouldn't be surjective. Thus the fusion rule is surjective. Now because $\widehat{V}_{n,l}$ is an extension of simple currents and fusion is right exact, the number of irreducible composition factors must not exceed two times the number of irreducible composition factors of $\widehat{V}_{n',l',x}^t$, and so by counting such factors one sees that this is an isomorphism.

Precomposing this with the embedding $\widehat{A}_{n+2\epsilon(l),l}\times \widehat{V}_{n',l',x}^t\to \widehat{V}_{n,l}\times  \widehat{V}_{n',l',x}^t$ and using that fusion with simple currents $\widehat{A}_{n\pm 1,l}$ preserve socle series, one gets the desired image.

\end{proof}

\newtheorem{fusextensionVbad}[DefKL]{Lemma}

\begin{fusextensionVbad}\label{fusextensionVbad}
Let $l \in \mathbb{Z}$. Then one has:
\begin{equation}
\widehat{A}_{n,l}\times \widehat{V}_{n',-l,x}^t=\left\{
    \begin{array}{rl}
     \widehat{V}_{n+n'-\frac{1}{2}, 0,-}^t & \text{if } l\geq 0,\\
      \\
       \widehat{V}_{n+n+\frac{1}{2}, 0, +}^t& \text{if } l<0.
    \end{array} \right.
\end{equation}

\end{fusextensionVbad}

\begin{proof}
One has the following intertwiner from the isomorphism $V_{n,l}\otimes V_{n',-l,x}^t\to V_{n,l}\otimes V_{n',-l,x}^t$: 
\begin{equation}
    \map{\mathcal{Y}}{\widehat{V}_{n,l}\times \widehat{V}_{n',-l,x}^t}{\widehat{V_{n,l}\otimes V_{n',-l,x}^t}}.
\end{equation} Since now $V_{n,l}\otimes V_{n',-l,x}^t\cong P_{n+n',0}^t$ belongs to $S_g$, one can again show that this $\mathcal{Y}$ must be an isomorphism.

When $l\geq  0$, the composition factors of $\widehat{A}_{n,l}\times \widehat{V}_{n',-l,x,-}^t$ are $\widehat{A}_{n+n'-1,0}$ and $\widehat{A}_{n+n',0}$ and has a unique minimal submodule $\widehat{A}_{n+n'-1,0}$. The only quotient of $\widehat{P}_{n+n',0}^t$ having this property and the right number of irreducible factors is $\widehat{V}_{n+n'-\frac{1}{2},0,-}^t$.

When $l< 0$ the composition factors of $\widehat{A}_{n,l}\times \widehat{V}_{n',-l,x,-}^t$ are $\widehat{A}_{n+n'+1,0}$ and $\widehat{A}_{n+n',0}$ and has a unique minimal submodule $\widehat{A}_{n+n'+1,0}$. The only quotient of $\widehat{P}_{n+n',0}^t$ having this property and the right number of irreducible factors is $\widehat{V}_{n+n'+\frac{1}{2},0,+}^t$. 
This completes the proof. 

\end{proof}

Using this, we obtain the following fusion rules:

\newtheorem{diamondfus}[DefKL]{Corollary}

\begin{diamondfus}
For $l,l'\in \mathbb{Z}$ and $m\geq m'$, one has the following fusion rule: \begin{equation}\begin{aligned}\widehat{P}_{n,l,x}^m\times \widehat{P}_{n',l', x}^{m'}= &\bigoplus_{0\leq t\leq m'-1}\widehat{P}_{n+n'+1-\epsilon(l,l'),l+l',x}^{m+m'-1-2t}\\ &\bigoplus_{0\leq t\leq m'-1} 2\, \widehat{P}_{n+n'-\epsilon(l,l'),l+l',x}^{m+m'-1-2t}\\ &\bigoplus_{0\leq t\leq m'-1} \widehat{P}_{n+n'-1-\epsilon(l,l'),l+l',x}^{m+m'-1-2t}\end{aligned}\end{equation}

For $l>0$, one has:
\begin{equation}
\widehat{V}_{n,l,x}^m\times \widehat{V}_{n',-l,x}^{m'}=\bigoplus_{0\leq k\leq m'-1} \widehat{P}_{n+n',0,x}^{m+m'-1-2k}
\end{equation}

\end{diamondfus}

\begin{proof}
When $l=0$, these follow from the tensor product structure of $\gl$. For $l\ne 0$, we use the fact that $\widehat{P}_{n,l,x}^m=\widehat{A}_{n,l}\times \widehat{P}_{0,0,x}^m$.
\end{proof}

\section{The mirror symmetry statement}\label{4}

In this section, we prove Theorem \ref{LALBequiv}. In Section \ref{4.1} we introduce the category $KL^0$, the de-equivariantization $KL^0/\mathbb{Z}$ as well as the lifting functor $\mathcal{F}$ from $KL^0$ to the category of modules of the VOSA $V_{\beta\gamma}\otimes V_{bc}$. The main ingredient here is a simple current extension described in \cite{creutzig2013w}, as well as the machinery of \cite{creutzig2017tensor, creutzig2020direct}. In Section \ref{4.2}, we recall free-field realizations of the VOSAs of our interest, following \cite{creutzig2020duality, aw}, and show that they are compatible with the simple current extension; we use them to show that certain objects in $KL$ do lie in $KL^0$, and so can be lifted by $\mathcal{F}$. As a consequence, we show the first part of Theorem \ref{LALBequiv}, namely that $\cbg$ has the structure of a braided tensor category defined by $P(z)$-intertwiners. In Section \ref{4.3}, we use our knowledge of $\cbg$ and $KL^0$ established in Section \ref{2} and Section \ref{3} to complete the proof of Theorem \ref{LALBequiv}. In Section \ref{4.4}, we compute fusion rules of indecomposable objects in $\cbg$ using the tensor equivalence $\mc F$. 

\subsection{$\beta\gamma$ as a simple current extension}\label{4.1}

As we have discussed in the introduction, due to the contribution of monopole operators, the physical boundary VOSA for $T_B$ is a simple current extension of $V(\widehat \gl)$. In \cite{creutzig2013w}, the authors showed that the VOSA $V(\widehat \gl)$ has many simple current extensions, which can often be identified as known VOSAs. A simple current extension is a direct sum of simple currents of $V(\widehat \gl)$ such that the resulting module carries a VOSA structure that extends the VOSA structure of $V(\widehat \gl)$. Among all the simple current extensions, the following is the one we  consider:
\begin{equation}\label{simplecurrentextension}
    \widehat{A}_{0,0}\oplus \bigoplus_{m>0} \widehat{A}_{\frac{1-m}{2},m}\oplus \widehat{A}_{\frac{m-1}{2},-m}.
\end{equation}
The choice is made by comparing the indices. The index of this module is computed in \cite{creutzig2013w}:
\begin{equation}
     \sum\limits_{m\in \mathbb{Z}} q^{m^2/2}s^{m/2}\prod\limits_{i=1}^\infty \frac{(1+sq^{i+m})(1+s^{-1}q^{i-m})}{(1-q^i)^2}. 
\end{equation}
Here the power of $s$ records the weights under the gauge group generator $N_0$ and the power of $q$ records the conformal weights. This correctly reproduces the index computed in \cite[Equation 3.31]{dimofte2018dual}.\footnote{Note that the boundary condition has effective level $k_{\mathrm{eff}}=1$ since each hypermultiplet contributes to $\frac{1}{2}$ of the level.}

The module in equation \eqref{simplecurrentextension} has the structure of a vertex operator superalgebra, and this VOSA is isomorphic to $V_{\beta\gamma}\otimes V_{bc}$, the tensor product of the $\beta\gamma$ VOA with a pair of free complex fermions. Let us denote by ${\mc V}_{ext}$ this extended VOSA. 

The following is explained in \cite{creutzig2017tensor}: given a vertex operator superalgebra extension $A$ in a VOA module category $\mathcal C$ where $P(z)$-intertwiners define a symmetric monoidal structure on $\mathcal C$, the category of local modules of $A$ in $\mathcal C$ coincides with the category of generalized modules of the VOA $A$ as braided tensor supercategories. However, in our situation (equation \eqref{simplecurrentextension}) as well as in many other cases, the object $A$ does not live in $\mathcal C$ but in a suitable completion of $\mathcal C$. Thus one needs to take a completion of $\mathcal C$ to allow infinite direct sums. This is explained in \cite[Theorem 1.1]{creutzig2020direct}: under suitable circumstances, one can extend the symmetric monoidal structure from $\mathcal C$ to a completion called $\mathrm{Ind}(\mathcal C)$, such that the object $A$ is now contained in $\mathrm{Ind}(\mathcal C)$.  The authors then showed \cite[Theorem 1.4]{creutzig2020direct} that the category of generalized $A$-modules in $\mathcal C$ also has a braided tensor supercategory structure defined via $P(z)$-intertwiners (see Section \ref{2.4}).   

We  apply this method to $KL$ and $\cbg$. Denote now by $\mathrm{Ind}(KL)$ the completion of $KL$ in the sense of \cite{creutzig2020direct}, then ${\mc V}_{ext}$ is a VOSA object in $\mathrm{Ind}(KL)$. Denote by $\mathrm{Rep}^0({\mc V}_{ext})$ the category of generalized modules of ${\mc V}_{ext}$ that lie in $\mathrm{Ind}(KL)$. This is a braided tensor supercategory via the $P(z)$-intertwining maps, as was shown in \cite[Theorem 1.4]{creutzig2020direct}.  

Given such an extension, for any object $W$ inside $\mathrm{Ind}(KL)$, the product ${\mc V}_{ext}\times W$ has the action by the mode algebra of the extended VOSA ${\mc V}_{ext}$. However, this action is not local in general. In \cite{creutzig2020simple}, the author explained that the monodromy determines whether the resulting action is local, thus becoming a generalized module of the VOSA ${\mc V}_{ext}$. More precisely, monodromy is defined by a composition of braiding isomorphisms: 
\begin{equation}
    M: {\mc V}_{ext}\times W\longrightarrow W\times {\mc V}_{ext}\longrightarrow {\mc V}_{ext}\times W.
\end{equation}
The module ${\mc V}_{ext}\times W$ is local if and only if the map $M$ is the identity morphism, or in other words, the monodromy acts trivially. Let us denote by $KL^0$ the full tensor subcategory of $KL$ consisting of $W$ such that monodromy on ${\mc V}_{ext}\times W$ is trivial. One has a functor: 
\begin{equation}
\map{\mathcal{F}}{KL^0}{\mathrm{Rep}^0({\mc V}_{ext})}
\end{equation}
that maps an object $W$ to ${\mc V}_{ext}\times W$ as an object in $\mathrm{Ind}(KL)$ together with the natural structure of a ${\mc V}_{ext}$ module. The result of \cite{creutzig2020simple} immediately implies:

\newtheorem{Liftmonodromy}{Theorem}[section]

\begin{Liftmonodromy}\label{Liftmonodromy}
 The functor:
\begin{equation}
  \map{\mathcal{F}}{KL^0}{\mathrm{Rep}^0({\mc V}_{ext})}
\end{equation}
is a tensor functor between braided tensor supercategories. 
\end{Liftmonodromy}

The simple currents $\widehat{A}_{0,\pm 1}$ generate a tensor subcategory of $KL^0$ isomorphic to $\mathrm{Rep}(\mathbb{C}^*)$. This tensor subcategory lives in the center of $KL^0$. We  denote by $KL^0/\mathbb{Z}$ the de-equivariantization of $KL^0$ by $\mathrm{Rep}(\mathbb{C}^*)$ in the sense of \cite[Theorem 8.23.3]{egno}. The functor $\mc F$ identifies the image of $KL^0$ with the de-equivariantization $KL^0/\mathbb{Z}$. The notation comes from the isomorphism $\mathrm{Rep}(\mathbb{C}^*)\cong \mathrm{Coh}(\mathbb{Z})$, and the action of $\mathrm{Rep}(\mathbb{C}^*)$ on $KL^0$ can be understood as an action of $\mathbb{Z}$ on the objects of $KL^0$. The de-equivariantization can be understood as the quotient category, where objects that are related by $\mathbb{Z}$ are considered equivalent. There is an action of $\mathrm{Coh}(\mathbb{Z}_2)$ on $KL^0/\mathbb{Z}$ generated by parity shift, and the de-equivariantization $KL^0/(\mathbb{Z}\times \mathbb{Z}_2)$ is the category of line operators for $T_B$ discussed in the introduction.\footnote{We note that de-equivariantization by $\mathbb{Z}_2$ does not change much of the category, apart from identifying an object with its parity shift, and forgetting the parity of Hom spaces. In particular, no new homomorphism is introduced since the parity is generated by an internal symmetry of $V(\widehat \gl)$.}

The category $ \mathrm{Rep}^0({\mc V}_{ext})$ may seem abstract, or at least not immediately related to $\cbg$. However, since ${\mc V}_{ext}\cong V_{\beta\gamma}\otimes V_{bc}$, and the category of generalized modules for $V_{bc}$ is equivalent to the category of super vector spaces, this category $ \mathrm{Rep}^0({\mc V}_{ext})$ is the Deligne product of a category of modules of the VOA $V_{\beta\gamma}$ with $\mathrm{SVect}$, the category of super vector spaces. What we  show in the following sections is that the image of $\mathcal{F}$ is identified with $\cbg\boxtimes \mathrm{SVect}$. The difficulty of analyzing this functor lies in the fact that computation of monodromy is complicated; it lacks concrete algebraic expressions. One way to go around this difficulty is through free-field realizations. This method is based on the observation that lifting modules from one free-field VOA to another is much easier to deal with. We  now go on to introduce free-field realizations of $V(\widehat \gl)$ as well as of ${\mc V}_{ext}$ that are compatible with the embedding $V(\widehat \gl)\to {\mc V}_{ext}$. In what follows, when we compute fusion product of objects that are infinite direct sums, we  always mean the fusion product in the sense of  \cite{creutzig2020direct}.

\subsection{Free-field realizations}\label{4.2}

\subsubsection{A free-field realization of $V(\widehat \gl)$}

We start with describing the free-field realization of $V(\widehat \gl)$ given in \cite{creutzig2020duality}. Let $X,Y,Z$ be a triple of free bosons with the non-degenerate pairing ${(X,Y)=(Z,Z)=1}$. The Heisenberg VOA generated by this triple is denote by $\mathbb{F}_0$. For any linear combination of $X,Y,Z$, say $\mu=aX+bY+cZ$, there is a simple module, denoted by $\mathbb{F}_\mu$, generated by the vacuum vector $\vert \mu \rangle$ whose weights under $\partial X,\partial Y$  and $\partial Z$ are given by the non-degenerate pairing between $\mu$ and $X,Y,Z$. The module $\bigoplus_{n\in \mathbb{Z}} \mathbb{F}_{nZ}$ has the structure of a VOSA, generated by $\mathbb{F}_0$ together with the vertex operators $\normord{e^{\pm Z(z)}}$, defined by:
\begin{equation}
    \normord{e^{\pm Z(z)}} = e^{\pm Z}z^{\pm Z_0} \prod_{m\geq 1} \exp \!\left(\pm \frac{Z_{-m}}{m}z^m\right)\exp\!\left(\pm \frac{-Z_m}{m}z^{-m}\right).
\end{equation}
See the definition in \cite{frenkel2004vertex} Section 5.2. We treat the modes of $\normord{e^{\pm Z(z)}}$ as odd. This VOSA is denoted by $V_Z$. Through the Bose-Fermi correspondence (see \cite{frenkel2004vertex}, Section 5.3), the sub-algebra generated by $\normord{e^{\pm Z(z)}}$ is isomorphic to the $bc$ ghost VOSA $V_{bc}$, generated by two fermionic (odd degree) fields $b,c$ with OPE:
\begin{equation}
b(z)c(z)\sim \frac{1}{z-w}.    
\end{equation} 
Under this correspondence, $\normord{e^{Z(z)}}\mapsto b(z)$ and $\normord{e^{-Z(z)}}\mapsto c(z)$. 

There is an embedding of VOSAs $V(\widehat \gl)\to V_Z$ given by 
\begin{equation}\begin{aligned}& E(z)\mapsto \partial Y(z),~~~ N(z)\mapsto -\normord{c(z)b(z)}+ \partial X(z)-\frac{\partial Y(z)}{2},\\ &\psi^+(z)\mapsto b(z),~~~~ \psi^-(z)\mapsto c(z)\partial Y(z)+\partial c(z)\end{aligned}
\end{equation}

For each linear combination of $X$ and $Y$, say $\nu=aX+bY$, one obtain a simple module of $V_Z$ whose underlying object is:
\begin{equation}
V_{\nu, Z}:=V_Z\times_{\mathbb{F}_0} \mathbb{F}_\nu.
\end{equation}
On the other hand, if $\nu$ involves $cZ$ for $c\notin \mathbb{Z}$, the resulting module is not local with respect to $\normord{e^Z}$. This is a special example of Theorem \ref{Liftmonodromy}. We  now introduce the screening operator, which helps us understand the embedding $V(\widehat \gl)\to V_Z$:

\newtheorem{screen}[Liftmonodromy]{Definition}

\begin{screen}
Let us define an intertwiner $\map{S(z)}{V_{nY, Z}}{V_{(n-1)Y, Z}}$ by \begin{equation}S(z)= \normord{e^{Z(z)-Y(z)}} \end{equation} The screening operator is then defined as the residue: \begin{equation}S=\oint S(z)\mathrm{d}z.\end{equation}

\end{screen}

It is shown in \cite{creutzig2020duality} that $V(\widehat \gl)$ is the kernel of $\map{S}{V_{0, Z}}{V_{-Y,Z}}$. The following Proposition shows how we can identify modules of $V(\widehat \gl)$ as restrictions of modules of $V_Z$:

\newtheorem{freefieldresmod}[Liftmonodromy]{Proposition}

\begin{freefieldresmod}\label{freefieldresmod}

Let $\mu=aX+bY$. When $a\in \mathbb{Z}$, $V_{\mu,Z}\cong \widehat{V}_{(2b-a-1)/2, a,-}$ as $V(\widehat \gl)$ modules. When $a\notin\mathbb{Z}$, $V_{\mu,Z}\cong \widehat{V}_{(2b-a-1)/2, a}$ as $V(\widehat \gl)$ modules. 

\end{freefieldresmod}

\begin{proof}
The case when $a\in S_g$ is proven in \cite{creutzig2020duality}. For $a\in S_b$, we  compute the conformal weight of the kernel of the screening operator using \cite{creutzig2020ribbon}. To do so, we need to re-write the free-field realization using the following vectors: $\alpha=Y-Z,\beta=Z+X$ and $\gamma=Z+X-Y$. These are three orthogonal generators of the lattice. The screening operator corresponds to $\alpha_-=-\alpha$, and $\alpha_+=2\alpha$. In terms of these generators, the conformal weight vector is:
\begin{equation}
\frac{1}{2}(\alpha(-1)^2+\beta(-1)^2-\gamma(-1)^2)+\frac{\alpha(-2)}{2},
\end{equation}
and $X,Y,Z$ can be written as: \begin{equation}
X=\gamma+\alpha_+/2,~~Y=\beta-\gamma,~~Z=\beta-\gamma-\alpha_+/2.
\end{equation}
So for a general $\mu=aX+bY$, the module $V_{\mu, Z}$ can be written in terms of the generators $\alpha,\beta$ and $\gamma$ as:
\begin{equation}
\bigoplus_{m}\mathbb{F}_0 \left \vert \frac{a-m}{2}\alpha_++(b+m)\beta+(a-b-m)\gamma \right \rangle.
\end{equation}
For each $m$, the lowest conformal weight of the kernel of the screening operator restricted to:
 \begin{equation}\mathbb{F}_0 \left \vert \frac{a-m}{2}\alpha_++(b+m)\beta+(a-b-m)\gamma \right \rangle\end{equation}
 is given by $h=h_{\alpha_{1+m-a, 1}}+h_{(b+m)\beta}-h_{(a-b-m)\gamma}$, where: 
\begin{equation}\begin{aligned} & h_{(b+m)\beta}=\frac{(b+m)^2}{2}\\ & h_{(a-b-m)\gamma}=\frac{(a-b-m)^2}{2}\end{aligned}\end{equation}
and as in \cite{creutzig2020ribbon}:  
\begin{equation}
  h_{\alpha_{1+m-a, 1}}= \left\{
    \begin{array}{rl}
      \frac{(1+m-a)^2-(1+m-a)}{2} & \text{if } a\leq m,\\
      \frac{(1+a-m)^2-(1+a-m)}{2} & \text{if } a>m.
    \end{array} \right.
\end{equation}
Combining these we have an explicit formula for $h$:
\begin{equation}h= \left\{
    \begin{array}{rl}
      \frac{1}{2}(m^2+m+2ab-a) & \text{if } a\leq m,\\
      \frac{1}{2}(m^2-m+2ab+a) & \text{if } a>m.
    \end{array} \right.\end{equation}

Now suppose $a<0$, then if $m<a$, $m$ cannot be $1/2$, so the minimum of $h$ in this region is obtained at $m=a-1$, giving $\frac{1}{2}(a^2+2ab-2a+2)$; when $m\geq a$, then the minimum can be taken at $m=0$, which is $\frac{1}{2}(2ab-a)$. Since $a<0$, $\frac{1}{2}(2ab-a)$ is smaller. This is exactly the minimum conformal weight of $\widehat{A}_{\frac{2b-a-1}{2}, a}$. When $a>0$, one can show that the minimum conformal weight is $\frac{1}{2}(2ab+a)$ which is the minimum conformal weight of $\widehat{A}_{\frac{2b-a+1}{2}, a}$. We thus have embedding $\widehat{A}_{\frac{2b-a-1}{2}, a}\subset V_{\mu, Z}$ when $a<0$ and $\widehat{A}_{\frac{2b-a+1}{2}, a}\subset V_{\mu, Z}$ when $a>0$. 

To finish the proof, we utilize the free-field fusion rule $\mathbb{F}_0\times_{\mathbb{F}_0}\mathbb{F}_\mu\cong \mathbb{F}_\mu$, which after combining the embedding $\widehat{A}_{\frac{2b-a}{2}+\epsilon(a), a}\to V_{\mu, Z}$ gives a $V(\widehat \gl)$ intertwiner:
\begin{equation}
\label{eqintertwinerbcF0}V_Z\times \widehat{A}_{\frac{2b-a}{2}+\epsilon(a), a}\to V_{\mu, Z}.
\end{equation}
Since $V_Z$ is indecomposable as a $V(\widehat \gl)$ module and $\widehat{A}_{\frac{2b-a}{2}+\epsilon(a), a}$ is a simple current, by Proposition 2.5 of \cite{creutzig2019schur}, the fusion $V_Z\times \widehat{A}_{\frac{2b-a}{2}+\epsilon(a), a}$ is still indecomposable, and we only need to show that it maps isomorphically onto $V_{\mu, Z}$. If we restrict this intertwiner to the sub-module $\widehat{A}_{0,0}\subset V_Z$, this composition $\widehat{A}_{0,0}\times  \widehat{A}_{\frac{2b-a}{2}+\epsilon(a), a}\to V_{\mu, Z}$ is nothing but the action of the VOA $V(\widehat \gl)$ on $ \widehat{A}_{\frac{2b-a}{2}+\epsilon(a), a}$, and so it is an isomorphism when restricted to this minimal-submodule. However, since $\widehat{A}_{0,0}$ is the unique irreducible sub-module of $V_Z$, the intertwiner in equation \eqref{eqintertwinerbcF0} must be injective. Now we see that it is also surjective since $V_{\mu, Z}$ only has two simple currents in its composition series. Hence $V_{\mu, Z}$ is in-decomposable and is isomorphic to $V_Z\times \widehat{A}_{\frac{2b-a}{2}+\epsilon(a), a}$. Now using Lemma \ref{fusextensionVbad}, we obtain the desired result. 

\end{proof}

It is in fact possible to identify the chains $\widehat{V}_{n,e,x}^t$ for $e\notin\mathbb{Z}$ and $\widehat{V}_{n,e,-,x}^t$ for $e\in \mathbb{Z}$ as restrictions of modules of the VOSA $V_Z$. Let us consider the module of $V_Z$ generated by the vacuum vector $(mX+nY)^{t-1}\vert \mu\rangle$ such that $n-m(x+\frac{1}{2})=0$, which we denote by $V_{\mu,x, Z}^t$. This module is an iterated self-extension of $V_{\mu, Z}$, and the action of $\partial X$ and $\partial Y$  have nontrivial Jordan blocks. 

\newtheorem{extensionfreefield}[Liftmonodromy]{Proposition}

\begin{extensionfreefield}\label{extensionfreefield}
Let $\mu=aX+bY$. When $a\in \mathbb{Z}$, there is an isomorphism of $V(\widehat \gl)$ modules: $\widehat{V}_{(2b-a-1)/2, a, x,-}^{t}\cong V_{\mu,x, Z}^t$. When $a\notin\mathbb{Z}$,  there is an isomorphism of $V(\widehat \gl)$ modules: $\widehat{V}_{(2b-a-1)/2, a, x}^{t}\cong V_{\mu,x, Z}^t$.

\end{extensionfreefield}

\begin{proof}

For the case when $a\notin \mathbb{Z}$, we proceed in the same way as the proof in \cite{creutzig2020duality}: note that there is a map of $\gl$ modules $V_{b-(a-1)/2, a, x}^{t}\to V_{\mu, x, Z}^t[0]$, where $ V_{\mu, x, Z}^t[0]$ is the lowest conformal weight space. This morphism induces a morphism of $V(\widehat \gl)$ modules: $\widehat{V}_{b-(a-1)/2, a, x,-}^{t}\to V_{\mu, x , Z}^t$. This is an isomorphism on the unique minimal submodule of $\widehat{V}_{b-(a-1)/2, a, x}^{t}$, thus is an embedding. It is thus an isomorphism by counting the number of irreducible factors. Similar situation holds for $a=0$.

When $a\in S_b$, we have the free field logarithmic intertwiner $V_{\mu, Z}\times_{V_Z} V_{0,x, Z}^t \cong V_{\mu,x, Z}^t$. Restricting this to the submodule $ \widehat{A}_{\frac{2b-a}{2}+\epsilon(a), a}\subseteq V_{\mu, Z}$ one gets an intertwiner of $V(\widehat \gl)$ modules:
\begin{equation}
\widehat{A}_{\frac{2b-a}{2}+\epsilon(a), a}\times V_{0,x,Z}^t\to V_{\mu,x, Z}^t.
\end{equation}
It is injective because it is non-zero when restricted to the unique irreducible sub-module $ \widehat{A}_{\frac{2b-a}{2}+\epsilon(a), a}\times \widehat{A}_{0,0}$. It is then an isomorphism by counting the number of composition factors. Thus comparing this with Lemma \ref{fusextensionVbad} we obtain the desired result. 

\end{proof}

In conclusion, many modules of $V(\widehat \gl)$ can be identified with restrictions of modules of $V_Z$. We want to remark here that although only $\widehat{V}_{n,e,-,x}^t$ show up in the above consideration, we can also find $\widehat{V}_{n,e,+,x}^t$  by pre-composing the embedding $V(\widehat \gl)\to V_Z$ with the automorphism $w$ introduced in Section  \ref{3.2}. 

\subsubsection{A free-field realization of ${\mc V}_{ext}$.}

Let us return to the VOSA ${\mc V}_{ext}$. Our goal is to introduce a free-field realization of ${\mc V}_{ext}$ that is compatible with the free-field realization $V(\widehat \gl)\to V_Z$ described in the last section. Recall that ${\mc V}_{ext}\cong V_{\beta\gamma}\otimes V_{bc}$. It is well known (see for example \cite{aw}) that $V_{\beta\gamma}$ has a free-field realization by a lattice VOA generated by $\partial \psi, \partial\theta$ and $e^{\pm \psi+\theta}$, where $(\psi,\psi)=-(\theta,\theta)=1$. There is a screening operator $S=\normord{e^\psi}$ such that $V_{\beta\gamma}$ is the kernel of $S$. By Bose-Fermi correspondence, $V_{bc}$ is isomorphic to a lattice VOSA of a single free boson. To conform with the free-field realization of $V(\widehat \gl)$, we set $\psi=Z-Y$ and $\theta=Y-Z-X$, which means that $\psi+\theta=-X$. We again treat the modes of $\normord{e^{\pm Z}}$ as odd.  With this redefinition, $Z+X$  becomes an independent variable (it has zero pairing with $\psi, \theta$), whose associated lattice VOSA is $V_{bc}$. Thus we can extend $V_Z$ by the operator $\normord{e^{\pm X}}$, or in other words, the $V_Z$ module $\mathcal{V}=\bigoplus_{n\in \mathbb{Z}}V_{nX, Z}$ has the structure of a VOSA, and is the lattice VOSA associated to the lattice generated by $X,Z$. The embedding ${\mc V}_{ext}\to \mathcal{V}$ is given by:
\begin{equation}b\mapsto \normord{e^{Z+X}},~~~c\mapsto \normord{e^{-Z-X}},~~~ \beta\mapsto \normord{e^{-X}},~~~ \gamma\mapsto \normord{(\partial Z-\partial Y) e^{X}}\end{equation} 
and its image is the kernel of the screening operator $S=\oint_z \normord{e^{\psi}}=\oint_z \normord{e^{Z-Y}}$. We thus have:

\newtheorem{freefieldbg}[Liftmonodromy]{Proposition}

\begin{freefieldbg}
There are free-field realizations compatible with the embedding $V(\widehat \gl)\to {\mc V}_{ext}$ given by the following diagram:

\begin{equation}
\begin{tikzcd}
V(\widehat \gl) \rar \dar
    & V_Z\cong \bigoplus_n \mathbb{F}_{nZ} \dar \\
{\mc V}_{ext} \rar
    & {\mc V}\cong \bigoplus_{m,n}\mathbb{F}_{mX+nZ}
\end{tikzcd}
\end{equation}

\end{freefieldbg}

We now come back to the question of lifting modules via the functor ${\mc F}$. This can be done by lifting modules from $V_Z$ to ${\mc V}$, and then restricting them to ${\mc V}_{ext}$ modules. The modules of $V_Z$ are given by $V_{mX+nY,Z,x}^t$, and they are generated by the vacuum module $(mX+nY)^{t-1}\vert \mu \rangle$. However, not all of them can be lifted to ${\mc V}$. In order to lift to ${\mc V}$, $n=0$ otherwise the action of $\normord{e^X}$ is non-local. Since $n-m(x+\frac{1}{2})=0$, this requires that $x=-\frac{1}{2}$. Similarly, we cannot use any linear combination $\mu=aX+bY$, since otherwise the conformal weight of $\normord{e^{-X}}$ is not an integer. To have integer conformal weight, we need that $b\in \mathbb{Z}$. Thus we require $b\in \mathbb{Z}$ and $x=-\frac{1}{2}$. We denote by ${\mc V}_{\mu}^t$ the resulting module of ${\mc V}$. On the other hand, given a module of $\mc W$ of $V_\bg$, we can view the Deligne product $\mc W \boxtimes V_{bc}$ as a module of $\mc V_{ext}$ using the isomorphism ${\mc V}_{ext}\cong V_{\beta\gamma}\otimes V_{bc}$. Recall the $V_{\beta\gamma}$ modules ${\mc W}_{[a]}^t$ and $({\mc W}_0^\pm)^t$. When $t=1$, the following is shown in \cite[Proposition 2.12]{aw}:

\newtheorem{freefieldallen}[Liftmonodromy]{Proposition}

\begin{freefieldallen}
Let $\mu=aX+bY$ with $b\in \mathbb{Z}$. When $a\notin \mathbb{Z}$, there are isomorphisms of ${\mc V}_{ext}$ modules:
\begin{equation}
{\mc V}_{\mu}\cong \sigma^{-b+1}  {\mc W}_{[-a]}\boxtimes V_{bc}.
\end{equation}
When $a\in \mathbb{Z}$, there are isomorphism of ${\mc V}_{ext}$ modules:
\begin{equation}
{\mc V}_{\mu}\cong \sigma^{-b+1}  ({\mc W}_{[0]}^-)\boxtimes V_{bc}.
\end{equation}

\end{freefieldallen}

We can extend it to the following:

\newtheorem{freefieldbgmodule}[Liftmonodromy]{Proposition}

\begin{freefieldbgmodule}\label{freefieldbgmodule}
Let $\mu=aX+bY$ with $b\in \mathbb{Z}$. When $a\notin \mathbb{Z}$, there are isomorphisms of ${\mc V}_{ext}$ modules: 
\begin{equation}
{\mc V}_{\mu}^t\cong \sigma^{-b+1}  {\mc W}_{[-a]}^t\boxtimes V_{bc}. 
\end{equation}
When $a\in \mathbb{Z}$, there are isomorphism of ${\mc V}_{ext}$ modules:
\begin{equation}
{\mc V}_{\mu}^t\cong \sigma^{-b+1}  ({\mc W}_{[0]}^-)^t\boxtimes V_{bc}.
\end{equation}

\end{freefieldbgmodule}

\begin{proof}
When $t=1$, this is simply \cite[Proposition 2.12]{aw}. When $t>1$, we only need to show that the module ${\mc V}_{\mu}^t$ is indecomposable, and the result will follow from Theorem \ref{cbgindecomptyp} and Theorem \ref{cbgindecompatyp}: the only indecomposable modules having $\sigma^{-b+1}  {\mc W}_{[-a]}$ or $\sigma^{-b+1}  {\mc W}_{[0]}^-$ in its composition series are the long chains. To show that it is indecomposable, we compute the action of $\normord{\gamma\beta}$ on the generator $X^{t-1}\vert \mu \rangle$. It is easy to see that there is a Jordan block of size $t$, and thus it is indecomposable. 

\end{proof}

With this, we derive:

\newtheorem{bgmodffid}[Liftmonodromy]{Proposition}

\begin{bgmodffid}\label{bgmodffid}
Let $b\in \mathbb{Z}$. If $a\notin \mathbb{Z}$, one has the following isomorphism of ${\mc V}_{ext}$ modules:
\begin{equation}{\mc F}\left( \widehat{V}_{(2b-a-1)/2, a,x=-1/2}^t\right)\cong  \sigma^{-b+1}\mathcal{W}_{[-a]}^t\boxtimes V_{bc}.\end{equation} For $a\in \mathbb{Z}$, one has the following isomorphism of ${\mc V}_{ext}$ modules:
\begin{equation}
\begin{aligned}
&{\mc F}\left(\widehat{V}_{(2b-a-1)/2, a, -,x=-1/2}^t\right)\cong  \sigma^{-b+1} (\mathcal{W}_0^-)^t\boxtimes V_{bc},\\
&{\mc F} \left(\widehat{V}_{(2b-a-1)/2, a, +,x=-1/2}^t\right)\cong  \sigma^{-b} (\mathcal{W}_0^+)^t\boxtimes V_{bc}.
\end{aligned}
\end{equation}

\end{bgmodffid}

\begin{proof}

We have seen that ${\mc V}\times_{V_Z}V_{\mu, x=-1/2,Z}^t\cong {\mc V}_{\mu}^t$. The inclusion ${\mc V}_{ext}\to {\mc V}$ induces an isomorphism:
\begin{equation}
{\mc F}\left({\mc V}_{\mu}^t\right)\cong {\mc V}\times_{V_Z}V_{\mu, x=-1/2,Z}^t
\end{equation}
as modules of the mode algebra of ${\mc V}_{ext}$. Comparing this with Proposition \ref{extensionfreefield} and Proposition \ref{freefieldbgmodule}, we get the desired result for $ \widehat{V}_{b-(a-1)/2, a,x=-1/2}^t$ when $a\notin \mathbb{Z}$ and $ \widehat{V}_{b-\frac{a-1}{2}, a, -,x=-1/2}^t$ when $a\in \mathbb{Z}$. Twisting both free-field realizations of $V(\widehat \gl)$ and $V_{\beta\gamma}$ by conjugation $w$ gives the result for $\widehat{V}_{b-\frac{a-1}{2}, a, +,x=-1/2}^t$.

\end{proof}

\noindent\textbf{Remark.} The requirement that $x=-1/2$ implies that $N_0+E_0/2$ acts semi-simply, and requiring that $b\in \mathbb{Z}$ implies that it has integer eigenvalues. The above Proposition implies that a sufficient condition for a $V(\widehat \gl)$ module to be lifted to ${\mc V}_{ext}$ is that $N_0+E_0/2$ acts semi-simply with integer eigenvalues. We  denote by $KL^{N+E/2}$ the subcategory of $KL$ where  $N_0+E_0/2$ acts semi-simply with integer eigenvalues. 

\vspace{5pt}

As a corollary, we prove the first part of Theorem \ref{LALBequiv}:

\newtheorem{cbgtensor}[Liftmonodromy]{Corollary}

\begin{cbgtensor}
The category $\cbg$ has the structure of a braided tensor category defined by $P(z)$-intertwining operators. 
\end{cbgtensor}

\begin{proof}
As discussed in Section \ref{4.1}, the supercategory $\mathrm{Rep}^0(\mathcal{V}_{ext})$ has the structure of a braided tensor supercategory defined by $P(z)$-intertwining operators. By Theorem \ref{cbgindecomptyp}, Theorem \ref{cbgindecompatyp} as well as Proposition \ref{bgmodffid}, we conclude that $\cbg\boxtimes \mathrm{SVect}$ is a full subcategory of $\mathrm{Rep}^0(\mathcal{V}_{ext})$, so we only need to show that this subcategory is a closed under fusion product. By \cite{creutzig2020tensor}, we conclude that fusion product on $KL$, and consequently on $\mathrm{Rep}^0(\mathcal{V}_{ext})$, is exact. Thus by the definition of $\cbg$, we only need to show that given two modules in $\cbg\boxtimes \mathrm{SVect}$, fusion product of their composition factors are still in $\cbg\boxtimes \mathrm{SVect}$. This is computed in \cite{aw}, and we conclude that $\cbg\boxtimes \mathrm{SVect}$ is closed under fusion, and is thus a braided tensor subcategory. Thus $P(z)$-intertwining operators define a braided tensor category structure on $\cbg$.
\end{proof}

\subsection{The equivalence}\label{4.3}

In this section, we  complete the proof of Theorem \ref{LALBequiv}. Since the work of \cite{creutzig2017tensor, creutzig2020direct} already implies that the functor $\mc F$ is a braided tensor functor, we only need to show that the image of $KL^0$ under $\mc F$ coincides with $\cbg\boxtimes \mathrm{SVect}$ as abelian supercategories. The procedure for the proof is the following:
\begin{enumerate}
    \item Recall the category $KL^{N+E/2}$. We first conclude from Proposition \ref{bgmodffid} that $KL^{N+E/2}$ is a subcategory of $KL^0$.
    
    \item We then use Theorem \ref{cbgindecomptyp} and \ref{cbgindecompatyp} to conclude that the lifting functor $\mathcal{F}$ restricted to $KL^{N+E/2}$ is essentially surjective onto $\cbg\boxtimes \mathrm{SVect}$.  
    
    \item Finally, we conclude from the above two points that the functor ${\mc F}$ identifies $KL^0/\mathbb{Z}$ with $\cbg\boxtimes \mathrm{SVect}$, and that $KL^0$ coincides with with $KL^{N+E/2}$.
    
\end{enumerate}
 Let us start with the following:
 
\newtheorem{LemKLNKL0}[Liftmonodromy]{Lemma}

\begin{LemKLNKL0}\label{LemKLNKL0}
$KL^{N+E/2}$ is a subcategory of $KL^0$.

\end{LemKLNKL0}

\begin{proof}
We need to show that objects in $KL^{N+E/2}$ can be lifted to $\mathrm{Rep}^0({\mathcal V}_{ext})$. Consider first $e\notin \mathbb{Z}$ with $\widehat{M}\in KL_e$ for some $\gl$ module $M$, and suppose $\widehat{M}\in KL^{N+E/2}$. To show that $\widehat{M}$ can be lifted to a local module of ${\mc V}_{ext}$, we only need to show that it is a quotient of a module that can be lifted, since if the monodromy is trivial on a module, it is trivial on all quotients. This is guaranteed if $M$ is a quotient of a direct sum of $V_{n,e,x=-1/2}^t$ for various $n\in (1-e)/2+\mathbb{Z}$ and $t$. Without loss of generality one may assume that $M$ is generated by a single $m$ such that $(E-e)^t m=0$, and $(N+E/2)m=km$ for some $k\in \mathbb{Z}$. We may choose $k$ to be the largest such value so that $\psi^+m=0$. It is clear then that the module $M$ is spanned by the vectors coming from applying $\gl$ to $m$:
\begin{equation}
    m, \psi^-m, (E-e)m, (E-e)\psi^-m,\ldots, (E-e)^{t-1}m, (E-e)^{t-1}\psi^-m.
\end{equation}
From the analysis of Section \ref{3.1}, there is a surjection: 
\begin{equation}
V_{k-(e+1)/2, e,x=-1/2}^t\to M
\end{equation}
by mapping the highest weight vector of the generator to $m$. This is a well-defined map since the structure of $V_{k-(e+1)/2, e,x=-1/2}^t$ is determined by an element $v$ with $(N+E/2)v=kv$, $\psi^+v=0$ as well as $(E-e)$ has order $t$. This shows that $M$ is a quotient of $V_{k-(e+1)/2, e,x=-1/2}^t$, and so $\widehat{V}_{k-(e+1)/2, e,x=-1/2}^t$ maps onto $\widehat{M}$. Since $\widehat{V}_{k-(e+1)/2, e,x=-1/2}^t$ can be lifted to a local module, so can $\widehat{M}$.

A similar argument can be applied when $e=0$ since any $M$ such that $N$ acts semi-simply with integer eigenvalues and $E$ acts nilpotently is a quotient of $P_{n,0,x=-1/2}^t$. 

Finally, let us consider when $e\in S_b$. Let $W\in KL_e$, then $W=\widehat{A}_{-e/2+\epsilon(e),e}\times \widehat{M}$ for some $M$. This means that $W$ can be extended if and only if $\widehat{M}$ can be. We must show that $N+E/2$ acts semi-simply on $M$ with integer eigenvalues. If this were not the case, two things can go wrong: first, $N+E/2$ does not have integer eigenvalues, or it does not act semi-simply. Suppose it was the first case, which means that $M$ has a sub-quotient $A_{n,0}$ such that $n\notin\mathbb{Z}$, but this means that $W$ has a sub-quotient $\widehat{A}_{-e/2+\epsilon(e),e}\times \widehat{A}_{n,0}=\widehat{A}_{n-e/2+\epsilon(e) ,e}$, a sub-module of $\widehat{V}_{n-e/2-\epsilon(e),e}$ on which $N_0+E_0/2$ acts semi-simply with eigenvalues in $n+\mathbb{Z}\ne \mathbb{Z}$, a contradiction to our assumption on $W$. If $N+E/2$ does not act semi-simply, then since $\widehat{A}_{-e/2+\epsilon(e) ,e}=\sigma^e(\widehat{A}_{e,0})$, one has 
\begin{equation}
    \widehat{A}_{-e/2+\epsilon(e),e}\times \widehat{M}\cong \sigma^e\left(\widehat{A_{e}\otimes M} \right).
\end{equation}
By the definition of spectral flow, the action of $N_0+E_0/2$ on $ \sigma^e\left(\widehat{A_{e}\otimes M} \right)$ is the same as the action of $N_0+E_0/2+e_0/2$ on $\widehat{A_{e}\otimes M}$. If $N+E/2$ does not act semi-simply on $M$, then $N_0+E_0/2+e_0/2$ can not act semi-simply since they only differ by a scalar, giving a contradiction to our assumption on $W$.

This means that we may apply the previous argument to $M$, so that $\widehat{M}$ can be lifted, and thus so can $W$. This completes the proof.

\end{proof}

Next, we prove:

\newtheorem{CorCbgKLN}[Liftmonodromy]{Proposition}

\begin{CorCbgKLN}\label{CorCbgKLN}
The functor $\mathcal{F}$ restricted to $KL^{N+E/2}$ is essentially surjective onto $\cbg\boxtimes \mathrm{SVect}$.
\end{CorCbgKLN}

\begin{proof}
To show that it is essentially surjective, by Theorem \ref{cbgindecomptyp} and \ref{cbgindecompatyp}, we need only show that the image of $\mathcal{F}$ contains $\sigma^l {\mc W}_{[a]}^t$ for $a\notin \mathbb{Z}$ and $\sigma^l ({\mc W}_{[0]}^\pm)^t$. These then follow from Proposition \ref{bgmodffid}. 

\end{proof}

Now we can finish the proof of Theorem \ref{LALBequiv}:

\newtheorem{Thmagain}[Liftmonodromy]{Theorem}

\begin{Thmagain}
The functor $\mathcal F$ provides an equivalence of braided tensor supercategories:
\begin{equation}
\cbg\boxtimes \mathrm{SVect}\cong KL^0/\mathbb{Z}. 
\end{equation}
\end{Thmagain}

\begin{proof}

Let us show that the image of $KL^0$ under $\mathcal{F}$ is $\cbg\boxtimes \mathrm{SVect}$. For this, it is enough to show that the simple modules of $KL^0$ coincides with those of $KL^{N+E/2}$. First of all, consider $e\notin \mathbb{Z}$, in which case the simple modules are $\widehat{V}_{n,e}$. If monomdromy acts trivially on ${\mc V}_{ext}\times_{V(\widehat \gl)} \widehat{V}_{n,e}$, then the monodromy of ${\mc V}\times_{V_Z}\widehat{V}_{n,e}$ would also be trivial. We have seen that this only happens if $n+e/2\pm \frac{1}{2}\in \mathbb{Z}$, or in other words, when $\widehat{V}_{n,e}\in KL^{N+E/2}$. When $e\in \mathbb{Z}$, the simple modules are $\widehat{A}_{n,e}$, and we embed them into $\widehat{V}_{n,e,-}$. Since monodromy can be computed by using the twist element $e^{2\pi i L_0}$, and $L_0$ acts semi-simply on $\widehat{V}_{n,e,-}, {\mc V}_{ext}$ as well as the fusion ${\mc V}_{ext}\times_{V(\widehat \gl)}  \widehat{V}_{n,e,-}$, we know that the monodromy map is semi-simple. More-over, $\widehat{V}_{n,e,-}$ is indecomposable as a $V(\widehat \gl)$ module, and so the monodromy must be scalar on each direct summand of ${\mc V}_{ext}\times_{V(\widehat \gl)}  \widehat{V}_{n,e,-}$. This means that if monodromy action on ${\mc V}_{ext}\times_{V(\widehat \gl)}  \widehat{V}_{n,e,-}$ is non-trivial, then it is nontrivial on ${\mc V}_{ext}\times_{V(\widehat \gl)}  \widehat{A}_{n,e}$ by naturality of commutativity constraints and exactness of fusion. Using similar argument as in the case of $e\notin\mathbb{Z}$, this means that $\widehat{V}_{n,e,-}$ is an object in $KL^{N+E/2}$, and so $\widehat{A}_{n,e}$ must also be in $KL^{N+E/2}$. 

By the discussions of Section \ref{4.1}, the functor $\mathcal F$ identifies the image of $KL^0$ with the de-equivariantization. Thus we have an equivalence:
\begin{equation}
     KL^0/\mathbb{Z}\cong \mathcal{C}_{\beta\gamma}\boxtimes \mathrm{SVect}
\end{equation}
as desired.

\end{proof}

\newtheorem{KLN+E/2bgate}[Liftmonodromy]{Corollary}

\begin{KLN+E/2bgate}
$KL^{N+E/2}$ coincides with $KL^0$. Moreover, for each $e\in \mathbb{C}$, $\mathcal F$ gives an equivalence between abelian supercategories:
\begin{equation}
 \mathcal{F}: KL^{N+E/2}_e=KL^0_e\cong \mathcal{C}_{\beta\gamma, [e]}\boxtimes \mathrm{SVect}
\end{equation}

\end{KLN+E/2bgate}

\begin{proof}
We first show that $KL^{N+E/2}$ coincides with $KL^0$. Let $W\in KL^0_e$ for some $e$, then ${\mathcal F}(W)$ is an element in $\mathcal{C}_{\beta\gamma,[e]}$. By Corollary \ref{CorCbgKLN}, there is an element $V\in KL_{e}^{N+E/2}$ such that:
\begin{equation}
    \mathcal{F}(V)\cong \mathcal{F}(W)
\end{equation}
This is an equivalence as modules of ${\mc V}_{ext}$, so must be an equivalence as modules between $V(\widehat \gl)$. Taking the generalized $E_0$ weight $e$ part, we must have $W\cong V$. Thus $KL^{N+E/2}$ coincides with $KL^0$.

We now show that $\mathcal F$ restricts to an equivalence of abelian supercategories:
\begin{equation}
 \mathcal{F}: KL^{N+E/2}_e=KL^0_e\cong \mathcal{C}_{\beta\gamma, [e]}\boxtimes \mathrm{SVect}
\end{equation}
It is essentially surjective from Corollary \ref{CorCbgKLN}, and we only need to show that it is fully-faithful. For any $W_1,W_2\in KL_{e}^{N+E/2}$, if $\map{f}{W_1}{W_2}$ is nonzero, because ${\mc V}_{ext}$ is a direct sum of simple currents, the associated map:
\begin{equation}
{\mathcal F}(W_1)\to \mathcal{F}(W_2)
\end{equation}
is non-zero. Thus the functor ${\mc F}$ is faithful. To show that it is full, take any $W_1,W_2\in KL_{e}^{N+E/2}$ and $f\in \mathrm{Hom}(\mathcal{F}(W_1),\mathcal{F}(W_2))$. It is a morphism as ${\mc V}_{ext}$ modules, so must be a morphism in $\mathrm{Ind}(KL)$. By restricting this morphism to the generalized $E_0$ weight $e$ part, one obtains a morphism between $W_1$ and $W_2$. This provides natural bijections, and so the functor $\mathcal{F}$ is fully faithful when restricted to $KL_e^{N+E/2}$. 

\end{proof}

\subsection{The fusion structure of $\cbg$}\label{4.4}

Using the equivalence as well as the fusion structure of $V(\widehat \gl)$ obtained in Section \ref{3.4}, we have the following fusion structure for $\beta\gamma$ VOA:
\newtheorem{Fusionbg}[Liftmonodromy]{Corollary}

\begin{Fusionbg}
For any $\mc M\in \cbg$, denote by ${\mc M}_n$ the spectral flow $\sigma^n (\mc M)$. We have the following fusion rules for $\beta\gamma$ modules:
\begin{equation}\label{eqfusionbg}
\begin{aligned}
& \mathcal{P}_n^t\times \mathcal{P}_m^s\cong \bigoplus_{l=0}^{\min\{t,s\}-1} P_{n+m-1}^{s+t-1-2l}\oplus 2P_{n+m}^{s+t-1-2l}\oplus P_{n+m+1}^{s+t-1-2l},\\ & \mathcal{W}_{[\lambda], n}^t\times \mathcal{W}_{[\mu],m}^s\cong  \bigoplus_{l=0}^{\min\{t,s\}-1} \mathcal{W}_{[\lambda+\mu],n+m}^{t+s-1-2l}\oplus \mathcal{W}_{[\lambda+\mu],n+m-1}^{t+s-1-2l}, \text{if } \lambda+\mu\notin\mathbb{Z},\\ &\mathcal{W}_{[\lambda], n}^t\times \mathcal{W}_{[-\lambda],m}^s\cong\mathcal{W}_{+, n}^t\times \mathcal{W}_{-,m}^s\cong \bigoplus_{l=0}^{\min\{t,s\}-1} \mathcal{P}_{m+n-1}^{t+s-1-2l}
\end{aligned}
\end{equation}
\end{Fusionbg}

\section{Quiver algebra and quantum group}\label{5}

In this section, we will focus on the subcategory of atypical modules ${\mc C}_{\beta\gamma, [0]}$. In Section \ref{5.1}, we will compute the endomorphism of the identity line operator, and compare the result with the computation in the smaller category studied in \cite{aw}. In Section \ref{5.2}, we show that ${\mc C}_{\beta\gamma, [0]}$ can be described as the category of modules of a quiver algebra, and show that this quiver algebra can be related to the quantum group $\overline{U}_q^H(\mathfrak{sl}(2))$.

\subsection{The category of atypical modules}\label{5.1}

Recall the following decomposition in equation \eqref{eqcbgblock}:
\begin{equation}
 \cbg = \bigoplus_{\lambda \in \mbb C / \mbb Z} \mc C_{\bg, \lambda}.
\end{equation}
When $\lambda=[0]$, the full subcategory ${\mc C}_{\beta\gamma, [0]}$ is called the category of atypical modules. This is a tensor subcategory of $\cbg$, and the tensor identity of $\cbg$, which is $\mathbbm{1}={\mc V}=V_{\beta\gamma}$, lies in ${\mc C}_{\beta\gamma, [0]}$. 

On the other hand, by Theorem \ref{LALBequiv}, there are equivalences of abelian categories:
\begin{equation}
{\mc C}_{\beta\gamma, [0]}\cong KL_0^{N+E/2}/\mathbb{Z}_2\cong {\mc C}_0^{N+E/2}/\mathbb{Z}_2,
\end{equation}
where recall that ${\mc C}_0^{N+E/2}$ is the category of finite-dimensional modules of $\mathfrak{gl}(1|1)$ such that $N+E/2$ acts semi-simply with integer eigenvalues and $E$ acts nilpotently. Now consider the category ${\mc C}_0^{N+E/2}$. In the following, we will view the action of $N+E/2$ as a $\mathbb{C}^*$ grading on the objects of ${\mc C}_0^{N+E/2}$. Denote by $A$ the algebra over $\mathbb{C}$ generated by $\psi^{\pm}$. The only relations the two generators satisfy are $(\psi^+)^2=(\psi^-)^2=0$. The algebra $A$ can be viewed as a super algebra if we view $\psi^\pm$ as odd elements. The adjoint action of $N+E/2$ gives a $\mathbb{C}^*$ grading under which $\psi^\pm$ have weight $\pm 1$. Denote by $A-\mathrm{Mod}^{\mathbb{C}^*}$ the category of finite-dimensional $\mathbb{C}^*$-equivariant modules of the (non-super) algebra $A$, then it is clear by definition, that we have the following equivalence of abelian categories:
\begin{equation}
{\mc C}^{N+E/2}/\mathbb{Z}_2\cong A-\mathrm{Mod}^{\mathbb{C}^*}.
\end{equation}
Indeed, the supercategory ${\mc C}^{N+E/2}$ is equivalent to the supercategory of $\mathbb{C}^*$-equivariant, finite-dimensional modules of the super algebra $A$. Since the superalgebra structure can be induced by the action of $\mathbb{C}^*$, the de-equivariantization simply forgets the superalgebra structure. In other words, a $\mathbb{C}^*$ equivariant module automatically has a compatible $\mathbb{Z}_2$ grading. From the above two equations we conclude:
\begin{equation}
{\mc C}_{\beta\gamma, [0]}\cong A-\mathrm{Mod}^{\mathbb{C}^*}_{nil}.
\end{equation}
Here $A-\mathrm{Mod}^{\mathbb{C}^*}_{nil}$ is the subcategory of $A-\mathrm{Mod}^{\mathbb{C}^*}$ where $E$ acts nilpotently. Under this isomorphism, the identity object $\mathbbm{1}$ corresponds to the trivial representation $\mathbb{C}$ of $A$, on which $\psi^\pm$ acts trivially. We are now ready to compute:

\newtheorem{extidentity}{Proposition}[section]

\begin{extidentity}
The derived endomorphism algebra of the identity line in $\cbg$ is~trivial:
\begin{equation}
\mathrm{End}_{D^b \cbg}(\mathbbm{1})\cong \mathbb{C}.
\end{equation}
Here $D^b \cbg$ is the bounded derived category of the abelian category $\cbg$. 
\end{extidentity}

\begin{proof}

From the above, we only need to show that:
\begin{equation}
\mathrm{End}_{D^b A-\mathrm{Mod}^{\mathbb{C}^*}}(\mathbb{C})\cong \mathbb{C}.
\end{equation}
The algebra $A$ is well-known to be a Koszul algebra, and its Koszul dual is $A^!=\mathbb{C}[x,y]/(xy)$, a commutative algebra of two variables $x,y$ satisfying $xy=0$. Thus, the derived endomorphism of $\mathbb{C}$ as an ungraded $A$ module is:
\begin{equation}
\mathrm{End}_{D^b A-\mathrm{Mod}}(\mathbb{C})\cong A^!=\mathbb{C}[x,y]/(xy).
\end{equation}
Here we need to use a projective resolution of $\mathbb{C}$, which technically do not belong to $D^b A-\mathrm{Mod}$. This is not an issue for us since $D^b A-\mathrm{Mod}$ is a full subcategory of the derived category of finitely-generated $A$ modules. The $\mathbb{C}^*$ action gives $x$ and $y$ weight $1$ and $-1$ respectively. Taking $\mathbb{C}^*$-invariants, we find:
\begin{equation}
\mathrm{End}_{D^b A-\mathrm{Mod}^{\mathbb{C}^*}}(\mathbb{C})\cong (A^!)^{\mathbb{C}^*}=\mathbb{C}.
\end{equation}
This completes the proof. 
\end{proof}

We have finally proved the claim that $\cbg$ produces the correct Coulomb branch! In contrast, the category of modules considered in \cite{aw} will not produce the correct answer for our setting. Indeed, the category considered in \cite{aw} can be obtained from $\cbg$ by requiring $J_0$ to act semi-simply. On the algebra $A$, this corresponds to a further restriction $E=0$. We are thus restricted to the Grassmann algebra of two variables $B=\mathbb{C}[\epsilon_1,\epsilon_2]$. The Koszul dual is $B^!=\mathbb{C}[x,y]$, and we have:
\begin{equation}\label{badEnd}
\mathrm{End}_{D^b B-\mathrm{Mod}^{\mathbb{C}^*}}(\mathbb{C})\cong (B^!)^{\mathbb{C}^*}=\mathbb{C}[xy].
\end{equation}
The computation above produces a commutative algebra generated by a single variable $xy$, instead of the trivial algebra ($\mbb C$) which we expect. The $xy$ in equation \eqref{badEnd} corresponds to a degree $2$ extension of $V_{\beta\gamma}$ by itself, in the form of the following exact sequence:
\begin{equation}
\begin{tikzcd}
0 \rar & {\mc V} \rar & {\mc W}_0^+  \rar & {\mc W}_0^-  \rar & {\mc V}  \rar & 0.
\end{tikzcd}
\end{equation}
This is trivialized in our category $\cbg$ due to the existence of the indecomposable module with the following Loewy diagram:
\begin{equation}
\begin{tikzcd}
{\mc V} \rar & \sigma^{-1}{\mc V}  \rar & {\mc V}.
\end{tikzcd}
\end{equation}
This is similar to the discussion in \cite{costello2019higgs}, where the authors showed that the choice of the category changes the resulting endomorphism algebra.

\subsection{A quiver description}\label{5.2}

\subsubsection{Constructing the quiver}

From the mirror symmetry statement, we have an equivalence 
\begin{equation}
    \mathcal{C}_{\beta\gamma, [0]}\boxtimes \mathrm{SVect}\cong KL^{N+E/2}_0\cong \mathcal{C}^{N+E/2}_0.
\end{equation} Let us analyze the category $ \mathcal{C}^{N+E/2}_0$. Since the action of $N+E/2$ is semi-simple, let $e_n$ be the operator of projection onto the eigenspace with eigenvalue $n\in \mathbb{Z}$. Let $\sigma_n=\psi^+\circ e_n$ and $\tau_n=\psi^-\circ e_n$, it is clear then $\psi^2=0$ now becomes $\sigma_n\sigma_{n-1}=\tau_n\tau_{n+1}=0$. These $e_n,\sigma_n$ and $\tau_n$ is the path algebra of the following quiver:
\begin{equation}
\begin{tikzcd}
\cdots & e_{n-1} \arrow[r,bend left=50, "\sigma_{n-1}"] & \arrow[l,bend left=50, "\tau_{n}"]  e_n \arrow[r,bend left=50, "\sigma_{n}"] & \arrow[l,bend left=50, "\tau_{n+1}" below] e_{n+1} & \cdots
\end{tikzcd}
\end{equation}
quotient by the quadratic relation $\sigma_n\sigma_{n-1}=\tau_n\tau_{n+1}=0$. Let this algebra be denoted by $\Lambda$. It is graded once we give $\sigma_n$ and $\tau_n$ degree $1$. From the above description, the supercategory $ \mathcal{C}^{N+E/2}_0$ is equivalent to the supercategory of graded finite-dimensional unitary modules of $\Lambda$. To get an ordinary category instead of a supercategory, we only need to choose a grading for each module. We say that a vector is even if it's in the image of $e_n$ for $n$ even. Once we do this, we find an equivalence between $\mathcal{C}_{\beta\gamma,[0]}$ with the category of finite-dimensional unitary modules of $\Lambda$ (as an ungraded algebra) on which $\sigma_n\tau_{n+1}$ and $\tau_{n}\sigma_{n-1}$ act nilpotently for all $n$. The last condition comes from that $E$ acts nilpotently for objects in $\mathcal{C}_0$. This category we denote by $\Lambda-\mathrm{mod}_{nil}\cong \mathcal{C}_{\beta\gamma,[0]}$. The interesting thing about this category is that it has three in principle distinct braided tensor category structure. One of them comes from the identification with objects in $ \mathcal{C}^{N+E/2}_0$, another from the identification with objects in ${\mc C}_{\beta\gamma, [0]}$. The third one is rather surprising, as it comes from a morphism from $\overline{U}_q^H(\mathfrak{sl}(2))$.

\subsubsection{Relation to quantum group}

Consider the unrolled restricted quantum group $\overline{U}_q^H(\mathfrak{sl}(2))$ at the fourth root of unity $q=i$. By this, we mean the algebra generated by $E,F,H,K^\pm$ with the relation:
\begin{equation}
    \begin{aligned}
& KK^{-1}=K^{-1}K=1,~~KE=-EK,~~KF=-FK,\\
& [H,E]=2E,~~[H,F]=-2F,~~[E,F]=\frac{K-K^{-1}}{2i},~~E^2=F^2=0.
    \end{aligned}
\end{equation}
When considering modules, we also consider modules that satisfies relation $K=q^H=e^{\pi i H/2}$, although this is not a well-defined relation in the algebra itself. 

\newtheorem{Lemquiverqg}[extidentity]{Lemma}

\begin{Lemquiverqg}
Let $\sigma=\sum_n \sigma_n$, $\tau=\sum_n\tau_n$ and $L=\sigma\tau+\tau\sigma$. Let $f(x)$ be the Taylor series of the function  $\frac{1-e^{-\pi i x}}{x}$. The assignment:
\begin{equation}\label{eqLemquiverqg}
    H\mapsto \left(\sum_n 2ne_n \right)+ L,~~K\mapsto e^{\frac{\pi i}{2}H},~~E\mapsto \sigma ,~~F\mapsto \frac{\tau}{2i} f(L) K
\end{equation}
gives a well-defined action of $\overline{U}_q^H(\mathfrak{sl}(2))$ on finite-dimensional unitary modules of $\Lambda$, and satisfies $K=q^H$.
\end{Lemquiverqg}

\begin{proof}
These give well-defined operators on any finite-dimensional unitary module of $\Lambda$, so we only need to show that this gives the right commutation relation, which is a simple algebraic check. We only show the relation $[E,F]=\frac{K-K^{-1}}{2i}$. We have that:
\begin{equation}
    [\sigma, \frac{\tau}{2i} f(L) K]=\frac{1}{2i}\{\sigma, \tau\}f(L)K=\frac{1}{2i}Lf(L)K.
\end{equation}
By definition, $Lf(L)=1-e^{-\pi i L}=1-e^{-\pi i H}=1-K^{-2}$, and so $Lf(L)K=K-K^{-1}$, thus the relation.

\end{proof}

\noindent\textbf{Remark.} It can be shown that when restricted to the category of atypical modules of $\overline{U}_q^H(\mathfrak{sl}(2))$, this gives an equivalence to $\Lambda-\mathrm{mod}_{nil}$, with a well-defined inverse assignment. This assignment should be, in some sense, similar to the induction functor.

\vspace{5pt}

It turns out, that one can transfer the structure of the braided tensor category from $\overline{U}_q^H(\mathfrak{sl}(2))$ to $\Lambda-\mathrm{mod}$. For example, the co-product $\Delta_t$(t for twisted) is given by:
\begin{equation}
\begin{aligned}&\Delta_t(e_n)=\sum_{p+q=n} e_p\otimes e_q,~~\Delta_{t}(\sigma)=\sigma\otimes 1+e^{\frac{\pi i}{2}H}\otimes \sigma,~~\Delta_t(L)=L\otimes 1+1\otimes L\\
&\Delta_t(\tau)=2i(1\otimes F+F\otimes e^{-\frac{\pi i}{2}H})(e^{-\frac{\pi i}{2}H}\otimes e^{-\frac{\pi i}{2}H}) \frac{1}{f(L)}(L\otimes 1+1\otimes L)\\ 
\end{aligned}
\end{equation}
where $\frac{1}{f(x)}$ should be understood as the Taylor series of the quotient. This is very asymmetric compared to the coproduct $\Delta$ from $\gl$:
\begin{equation}
\Delta(e_n)=\sum_{p+q=n}e_p\otimes e_q,~~\Delta(\sigma)=\sigma\otimes 1+e^{\pi i \sum_n ne_n} \otimes \sigma,~~\Delta(\tau)=\tau\otimes 1+e^{\pi i \sum_n ne_n}\otimes \tau.
\end{equation}
Under the co-product $\Delta$, the braiding is trivial, but the braiding under $\Delta_t$ is non-trivial:
\begin{equation}
    R=e^{\pi i H\otimes H/4 }(1+ 2i E\otimes F),
\end{equation}
and so there is a non-trivial twist:
\begin{equation}
    \theta=K(e^{-\pi i H^2/4}-2iKF e^{-\pi i H^2/4}E),
\end{equation}
where $H,E,F$ should be expressed using elements in $\Lambda$ as in equation \eqref{eqLemquiverqg}. These formulas can be found in \cite{ohtsuki2002quantum}. We expect that for any $M,N\in \Lambda-\mathrm{mod}$, the two co-product $M\otimes N$ and $M\otimes_t N$ are isomorphic as $\Lambda$ modules. However, we do not expect the two coproducts to give equivalent braided tensor structures. It remains a question whether the tensor structure $\Delta_t$ induced from $\overline{U}_q^H(\mathfrak{sl}(2))$ is equivalent to that coming from $\cbg$.

\newpage
\nocite{*}
\printbibliography

\end{document}